\documentclass[a4paper,onecolumn]{IEEEtran}

\usepackage[utf8]{inputenc}
\usepackage[T1]{fontenc}

\usepackage{amsfonts}
\usepackage{mathrsfs}
\usepackage{amsmath,amsthm,amssymb}
\usepackage[usenames,dvipsnames]{xcolor}
\usepackage[colorlinks=true,urlcolor=Mahogany,linkcolor=Mahogany,citecolor=Mahogany,plainpages=false,pdfpagelabels]{hyperref}

\usepackage[backend=biber,style=ieee,maxcitenames=4,maxalphanames=100,maxbibnames=100,isbn=false]{biblatex}

\usepackage{tikz}
\usepackage{pgfplots}
\usetikzlibrary{fit,calc}

\theoremstyle{plain}
\newtheorem{theorem}{Theorem}[section]
\newtheorem{lemma}[theorem]{Lemma}

\newtheorem{corollary}[theorem]{Corollary}
\newtheorem{proposition}[theorem]{Proposition}

\theoremstyle{definition}
\newtheorem{definition}[theorem]{Definition}

\newcommand*{\cI}{\mathcal{I}}

\newcommand*{\cM}{\mathcal{M}}

\newcommand*{\cT}{\mathcal{T}}

\newcommand*{\cX}{\mathcal{X}}
\newcommand*{\cY}{\mathcal{Y}}
\newcommand*{\cZ}{\mathcal{Z}}

\newcommand*{\eps}{\varepsilon}

\newcommand*{\id}{\mathrm{id}}
\newcommand*{\ident}{\mathrm{id}}
\newcommand*{\tr}{\mathrm{tr}}

\newcommand{\sket}[1]{{\ensuremath{\lvert#1\rangle}}}
\newcommand{\lket}[1]{{\ensuremath{\left\lvert#1\right\rangle}}}
\newcommand{\ket}[1]{\mathchoice{\lket{#1}}{\sket{#1}}{\sket{#1}}{\sket{#1}}}

\newcommand{\sbra}[1]{{\ensuremath{\langle#1\rvert}}}
\newcommand{\lbra}[1]{{\ensuremath{\left\langle#1\right\rvert}}}
\newcommand{\bra}[1]{\mathchoice{\lbra{#1}}{\sbra{#1}}{\sbra{#1}}{\sbra{#1}}}

\newcommand{\sbraket}[2]{{\ensuremath{\langle#1\rvert#2\rangle}}}
\newcommand{\lbraket}[2]{{\ensuremath{\left\langle#1\!\left\rvert\vphantom{#1}#2\right.\!\right\rangle}}}
\newcommand{\braket}[2]{\mathchoice{\lbraket{#1}{#2}}{\sbraket{#1}{#2}}{\sbraket{#1}{#2}}{\sbraket{#1}{#2}}}

\newcommand{\sketbra}[2]{{\ensuremath{\lvert #1\rangle\langle #2\rvert}}}
\newcommand{\lketbra}[2]{{\ensuremath{\left\lvert #1\middle\rangle\middle\langle #2\right\rvert}}}
\newcommand{\ketbra}[2]{\mathchoice{\lketbra{#1}{#2}}{\sketbra{#1}{#2}}{\sketbra{#1}{#2}}{\sketbra{#1}{#2}}}

\newcommand{\proj}[1]{\ketbra{#1}{#1}}

\newcommand{\Supp}{\mathrm{supp}}

\newcommand{\mbE}{\mathbb{E}}
\newcommand{\mbP}{\mathbb{P}}
\newcommand{\mbR}{\mathbb{R}}

\newcommand*{\freq}[1]{\mathsf{freq}(#1)}
\newcommand*{\Var}[1]{\mathsf{Var}(#1)}

\newcommand*{\Max}[1]{\mathsf{Max}(#1)}
\newcommand*{\MinSigma}[1]{\mathsf{Min}_{\Sigma}(#1)}
\newcommand*{\Min}[1]{\mathsf{Min}(#1)}
\newcommand*{\XOR}{\;\mathtt{XOR}\;}

\DeclareMathOperator{\Cov}{Cov}
\DeclareMathOperator{\Variance}{Var}

\tikzstyle{porte} = [fill=blue!20, draw]
\tikzstyle{etiquette} = [font=\scriptsize]

\begin{document}

\title{Entropy accumulation with improved \\ second-order term}

\date{\today}

\author{Frédéric~Dupuis
        and~Omar~Fawzi
        \thanks{F. Dupuis is with the Université de Lorraine, CNRS, Inria, LORIA, F-54000 Nancy, France.}
        \thanks{O. Fawzi is with the Université de Lyon, ENS de Lyon, CNRS, UCBL, LIP, F-69342, Lyon Cedex 07, France.}
    \thanks{Manuscript received XXXXX.}}

\markboth{IEEE Transactions on Information Theory}
{Dupuis \MakeLowercase{\textit{et al.}}: Entropy accumulation with improved second-order}

\maketitle

\begin{abstract}
    The entropy accumulation theorem~\cite{dfr16} states that the smooth min-entropy of an $n$-partite system $A = (A_1, \ldots, A_n)$ is lower-bounded by the sum of the von Neumann entropies of suitably chosen conditional states up to corrections that are sublinear in $n$. This theorem is particularly suited to proving the security of quantum cryptographic protocols, and in particular so-called device-independent protocols for randomness expansion and key distribution, where the devices can be built and preprogrammed by a malicious supplier~\cite{arv16}. However, while the bounds provided by this theorem are optimal in the first order, the second-order term is bounded more crudely, in such a way that the bounds deteriorate significantly when the theorem is applied directly to protocols where parameter estimation is done by sampling a small fraction of the positions, as is done in most QKD protocols. The objective of this paper is to improve this second-order sublinear term and remedy this problem. On the way, we prove various bounds on the divergence variance, which might be of independent interest.
  \end{abstract}



\section{Introduction}\label{sec:introduction}

There are many protocols in quantum cryptography, such as quantum key distribution, that work by generating randomness. Such protocols usually proceed as follows: we perform a basic subprotocol $n$ times (for example, sending a photon in a random polarization from Alice to Bob), we then gather statistics about the protocol run (for example, we compute the error rate from a randomly chosen sample of the rounds), and we then conclude that the final state contains a certain amount of randomness, which can then be processed further. Mathematical tools that can quantify the amount of randomness produced by quantum processes therefore constitute the centerpiece of many security proofs in quantum cryptography. The entropy accumulation theorem~\cite{dfr16} provides such a powerful tool that applies to a very general class of protocols, including device-independent protocols.

Informally, the main result of \cite{dfr16} is the following. Suppose we have an $n$-step quantum process like the one depicted in Figure~\ref{fig:eat}, in which we start with a bipartite state $\rho_{R_0 E}$ and the $R_0$ share of the state undergoes an $n$ step process specified by the quantum channels $\mathcal{M}_1$ to $\mathcal{M}_n$. At step $i$ of the process, two quantum systems $A_i$ and $B_i$ are produced, from which one can extract a \emph{classical} random variable $X_i$. The goal is then to bound the amount of randomness present in $A_1^n$ given $B_1^n$, conditioned on the string $X_1^n$ being in a certain set $\Omega$. The $X_i$'s are meant to represent the data we do statistics on, for example $X_i$ might tell us that there is an error at position $i$, and we want to condition on the observed error rate being below some threshold. Stated informally, the statement proven in \cite{dfr16} is then
\begin{equation}
    H_{\min}^{\varepsilon}(A_1^n|B_1^n E, X_1^n \in \Omega)_{\rho} \geqslant n \left( \inf_{q \in \Omega} f(q) \right) - \sqrt{n} c.
    \label{eqn:intro-eat-statement}
\end{equation}
Here, the smooth min-entropy $H_{\min}^{\varepsilon}$ represents the amount of extractable randomness (see Definition~\ref{def:smooth-min-max-entropy}), the \emph{tradeoff function} $f(q)$ quantifies the worst-case amount of entropy produced by one step of the process for an input state that is consistent with observing the statistics $q$, and $c$ is a number that depends on $\varepsilon$, the event $\Omega$ and the tradeoff function $f$ but not on $n$. One would then apply this theorem by replacing the $\mathcal{M}_i$'s by one step of the cryptographic protocol to obtain the desired bound. This is done, for example, in~\cite{arv16,ADFRV18} for device-independent randomness expansion and quantum key distribution.

\begin{figure}
    \begin{center}
    \begin{tikzpicture}[thick]
        \draw
            (0, 0) node[porte, minimum height=.7cm, minimum width=.7cm] (m1) {$\mathcal{M}_1$}
            ++(2, 0) node[porte, minimum height=.7cm, minimum width=.7cm] (m2) {$\mathcal{M}_2$}
            ++(2, 0) node (dotdotdot) {$\cdots$}
            ++(2, 0) node[porte, minimum height=.7cm, minimum width=.7cm] (mn) {$\mathcal{M}_n$}
            (m1) ++(-.5, -1.2) node[etiquette] (a1) {$A_1$}
            (m1) ++(.5, -1.2) node[etiquette] (b1) {$B_1$}
            (m2) ++(-.5, -1.2) node[etiquette] (a2) {$A_2$}
            (m2) ++(.5, -1.2) node[etiquette] (b2) {$B_2$}
            (mn) ++(-.5, -1.2) node[etiquette] (an) {$A_n$}
            (mn) ++(.5, -1.2) node[etiquette] (bn) {$B_n$}
            ;
       \draw 
           (m1) ++(-1.5, 0) edge[->] node[midway, above, etiquette] {$R_0$} (m1)
           (m1) edge[->] node[midway, above, etiquette] {$R_1$} (m2)
           (m2) edge[->] node[midway, above, etiquette] {$R_2$} (dotdotdot)
           (dotdotdot) edge[->] node[midway, above, etiquette] {$R_{n-1}$} (mn)
            (m1) edge[->] (a1)
            (m1) edge[->] (b1)
            (m2) edge[->] (a2)
            (m2) edge[->] (b2)
            (mn) edge[->] (an)
            (mn) edge[->] (bn)
            ;
        \draw
        (m1) to ++(-1.5, 0) to  ++(-.5, .5) node[left] {$\rho^0_{R_0 E}$} to ++(.5, .5) coordinate (topright) to node[midway, above, etiquette] {$E$} ([xshift=.5cm] topright -| mn.east) coordinate (rightwall)
            ;
        \draw[->]
            (mn) to (mn.center -| rightwall) node[above, etiquette] {$R_n$}
            ;
        \draw
            ($.5*(a1)+.5*(b1)$) ++(0, -.8) node[etiquette] (x1) {$X_1$}
            ($.5*(a2)+.5*(b2)$) ++(0, -.8) node[etiquette] (x2) {$X_2$}
            ($.5*(an)+.5*(bn)$) ++(0, -.8) node[etiquette] (xn) {$X_n$}
            ;
        \draw[->] (a1) to (x1);
        \draw[->] (a2) to (x2);
        \draw[->] (an) to (xn);
        \draw[->] (b1) to (x1);
        \draw[->] (b2) to (x2);
        \draw[->] (bn) to (xn);
    \end{tikzpicture}
    \end{center}
    \caption{Illustration of the type of process that the entropy accumulation theorem applies to.}
    \label{fig:eat}
\end{figure}

While this method yields optimal bounds in the first order, the second-order term which scales as $\sqrt{n}$ is bounded more crudely, and for some applications, this term can become dominant very quickly. This is particularly the case in applications which estimate the amount of entropy produced by testing a small fraction of the positions, which includes a large number of protocols of interest. The reason for this is that the value of $c$ in Equation \eqref{eqn:intro-eat-statement} is proportional to the gradient of $f$. Now, suppose that we have a protocol where we are testing positions with probability $O(1/n)$; in general this will make the gradient of $f$ proportional to $n$\footnote{Without getting into details, the tradeoff function $f$ often takes the form $f(p) = g(\frac{p(1)}{\gamma})$, where $p$ is a distribution on $\{0,1\}$ and $\gamma$ is the testing probability and $g$ is a fixed affine function. As such if $\gamma = O(\frac{1}{n})$, the gradient of $f$ is $\Omega(n)$. We refer the reader to~\cite{arv16,ADFRV18} or Section~\ref{sec:DIRE} of this paper for more details on this.} and therefore the second-order term will become $\Omega(n^{3/2})$ and overwhelm the first-order term. This is worse than we would expect: when we perform the analysis using conventional tools such as Chernoff-Hoeffding bounds in cases that are amenable to it, we obtain a much better scaling behavior, and in particular we still expect a non-trivial bound when the testing rate is $O(1/n)$. As a further indication that the second-order term can be improved, we also note that in~\cite[Appendix B]{arv16}, they resort to applying the entropy accumulation theorem to blocks rather than single rounds in order to obtain a good dependence on the testing rate.

The goal of this paper is therefore to improve the second-order term in \eqref{eqn:intro-eat-statement}. Analyzing second-order correction terms is already commonplace in information theory ever since the 60s, with the work of Volker Strassen~\cite{s62} who gave second-order bounds for hypothesis testing and channel coding. This topic has also seen a revival more recently~\cite{h08,ppv10,polyanskiy-thesis}. Quantum versions of such bounds have been proven as well since then; for example, Li~\cite{l14-3} and Tomamichel and Hayashi~\cite{th12} have shown a second-order expansion for quantum hypothesis testing, and \cite{th12} additionally gives second-order expansions for several other entropic quantities of interest. Other more recent developments can also be found in~\cite{tt14,tt15,leditzky-thesis,bg14,bdl16,dl15,dpr16}. 

Most of these results go one step further than we will in this paper, in that they pin down the $O(\sqrt{n})$ term \emph{exactly}, usually by employing some form of the Berry-Esseen theorem to a carefully designed classical random variable. Unfortunately, this approach seems to fail here, and we must resort to slightly weaker bounds that nevertheless give the right scaling behavior for protocols with infrequent sampling, and that are largely good enough in practice.

\paragraph{Paper organization: }In Section \ref{sec:preliminaries}, we give the notation used and some preliminary facts needed for the rest of the paper, including the Rényi entropy chain rule that powers the original entropy accumulation result in Section \ref{sec:chain}. Section \ref{sec:divergence-variance} then introduces the \emph{divergence variance} which governs the form of the second-order term, and discusses some of its properties. In Section \ref{sec:renyi-vs-von-neumann}, we present a new bound for the Rényi entropy in terms of the von Neumann entropy, and then apply it to the entropy accumulation theorem in Section \ref{sec:accumulation}, with specific bounds for the case of protocols with infrequent sampling in Section \ref{sec:infrequent-sampling}. We then compute finite-block-size bounds for the particular application of device-independent randomness expansion in Section \ref{sec:DIRE} and conclude with some open problems in Section \ref{sec:conclusion}.

\section{Preliminaries}\label{sec:preliminaries}

\subsection{Notation} \label{sec:notation}
In the table below, we summarize some of the notation used throughout the paper:
\begin{center}
    \begin{tabular}{|c|l|}
        \hline
        \emph{Symbol} & \multicolumn{1}{c|}{\emph{Definition}}\\
        \hline
        $A, B, C, \dots$ & Quantum systems, and their associated Hilbert spaces\\
        $\mathcal{L}(A,B)$ & Set of linear operators from $A$ to $B$\\
        $\mathcal{L}(A)$ & $\mathcal{L}(A,A)$\\
        $X_{AB}$ & Operator in $\mathcal{L}(A \otimes B)$\\
        $\mathcal{I}_{A}$ & Identity map from $\mathcal{L}(A)$ to itself\\        
        $\cM_{A \to B}$ & The subscript $A \to B$ is to indicate that $\cM$ is a linear map from $\mathcal{L}(A)$ to $\mathcal{L}(B)$ \\        
        $\mathrm{D}(A)$ & Set of normalized density operators on $A$\\
        $X_A \geqslant Y_A$ & $X_A - Y_A$ is positive semidefinite\\
        $A_i^j$ (with $j \geqslant i$) & Given $n$ systems $A_1,\dots,A_n$, this is a shorthand for $A_i,\dots,A_j$\\
        $A^n$ & Often used as shorthand for $A_1,\dots,A_n$\\
        $\log(x)$ & Logarithm of $x$ in base 2\\
        $\Variance(X)$ & Variance of the random variable $X$\\
        $\Cov(X,Y)$ & Covariance of the random variables $X$ and $Y$\\
        \hline
        $D_{\alpha}(\rho \| \sigma)$ & Sandwiched Rényi divergence (Definition~\ref{def:sandwiched-renyi-divergence})\\
        $D'_{\alpha}(\rho \| \sigma)$ & Petz Rényi divergence (Definition~\ref{def:petz-renyi-divergence})\\
        $H_{\alpha}(A|B)_{\rho}$ & $-D_{\alpha}(\rho_{AB} \| \ident_A \otimes \rho_B)$\\
        $H^{\uparrow}_{\alpha}(A|B)_{\rho}$ & $-\inf_{\sigma_B} D_{\alpha}(\rho_{AB} \| \ident_A \otimes \sigma_B)$\\
        $H'_{\alpha}(A|B)_{\rho}$ & $-D'_{\alpha}(\rho_{AB} \| \ident_A \otimes \rho_B)$\\
        $D_{\min}(\rho \| \sigma)$ & $D_{\frac{1}{2}}(\rho \| \sigma)$\\
        $D_{\max}(\rho \| \sigma)$ & $D_{\infty}(\rho \| \sigma)$\\
        $V(\cdot)$ & Various divergence variance measures; see Section~\ref{sec:divergence-variance}\\
        \hline
    \end{tabular}
\end{center}
   
\subsection{Entropic quantities}
The central mathematical tools used in this paper are entropic quantities, i.e.~various ways of quantifying the amount of uncertainty present in classical or quantum systems. In this section, we give definitions for the quantities that will play a role in our results.

\begin{definition}[Relative entropy]
    For any positive semidefinite operators $\rho$ and $\sigma$, the \emph{relative entropy} is defined as
    \[ D(\rho \| \sigma) = \begin{cases} \frac{1}{\tr[\rho]} \tr[\rho (\log \rho - \log \sigma)] & \text{ if } \Supp(\rho) \subseteq \Supp(\sigma)\\ \infty & \text{ otherwise} \end{cases}. \]
    \label{def:petz-divergence}
\end{definition}

\begin{definition}[von Neumann entropy]
    Let $\rho_{AB} \in \mathrm{D}(AB)$ be a bipartite density operator. Then, the \emph{conditional von Neumann entropy} is defined as
    \[ H(A|B)_{\rho} = -D(\rho_{AB} \| \ident_A \otimes \rho_B). \]
    \label{def:von-neumann-entropy}
\end{definition}

Our proofs heavily rely on two versions of the Rényi relative entropy: the one first introduced by Petz~\cite{p84}, and the ``sandwiched'' version introduced in~\cite{wwy13,mdsft13}. We define both of these here:

\begin{definition}[Sandwiched Rényi divergence]
    Let $\rho$ be a quantum state, let $\sigma$ be positive semidefinite, and let $\alpha \in [\frac{1}{2}, \infty]$. Then, the sandwiched Rényi divergence is defined as
    \begin{equation}\label{eqn:def-sandwiched-divergence-case1}
        D_{\alpha}(\rho \| \sigma) = \begin{cases} \frac{1}{\alpha-1} \log \tr\left[ \left( \sigma^{-\frac{\alpha'}{2}} \rho \sigma^{-\frac{\alpha'}{2}} \right)^{\alpha} \right] & \text{ if } \alpha < 1 \text{ or } \alpha > 1 \text{ and }\Supp(\rho) \subseteq \Supp(\sigma)\\ \log \inf\{ \lambda : \rho \leqslant \lambda \sigma \} & \text{ if } \alpha = \infty\\ D(\rho \| \sigma) & \text{ if } \alpha=1\\ \infty & \text{ otherwise,} \end{cases}
    \end{equation}
    where $\alpha' := \frac{\alpha - 1}{\alpha}$. Note $D_{\infty}$ is also referred to as $D_{\max}$ and $D_{\frac{1}{2}}$ as $D_{\min}$.
    \label{def:sandwiched-renyi-divergence}
\end{definition}

\begin{definition}[Petz Rényi divergence]
    Let $\rho$ be a quantum state, let $\sigma$ be positive semidefinite, and let $\alpha \in [0, 2]$. Then, the Petz Rényi divergence is defined as
    \begin{equation}\label{eqn:def-petz-divergence-case1}
        D'_{\alpha}(\rho \| \sigma) = \begin{cases} \frac{1}{\alpha-1} \log \tr\left[ \rho^{\alpha} \sigma^{1-\alpha} \right] & \text{ if } 0 < \alpha < 1 \text{ or } 1 < \alpha \leqslant 2 \text{ and } \Supp(\rho) \subseteq \Supp(\sigma)\\ -\log \tr[\Pi_{\Supp(\rho)} \sigma] & \text{ if } \alpha = 0\\ D(\rho \| \sigma) & \text{ if } \alpha = 1\\ \infty & \text{ otherwise,} \end{cases}
    \end{equation}
    where $\Pi_{\Supp(\rho)}$ is the projector on the support of $\rho$.
    \label{def:petz-renyi-divergence}
\end{definition}

These relative entropies can be used to define a conditional entropy:

\begin{definition}[Sandwiched Rényi conditional entropy] \label{def_sandwichedentropy}
    For any density operator $\rho_{A B}$ and for $\alpha \in [\frac{1}{2}, \infty]$ the \emph{sandwiched $\alpha$-R\'enyi entropy of $A$ conditioned on $B$} is defined as
  \begin{align*}
      H_\alpha(A|B)_{\rho} = -D_{\alpha}(\rho_{AB} \| \ident_A \otimes \rho_B).
  \end{align*}
  Note that we also refer to $H_{\infty}(A|B)_{\rho}$ as $H_{\min}(A|B)_{\rho | \rho}$.
\end{definition}

It turns out that there are multiple ways of defining conditional entropies from relative entropies. Another variant that will be needed in this work is the following:
\begin{definition} \label{def_sandwichedentropy_uparrow}
  For any density operator $\rho_{A B}$ and for $\alpha \in [\frac{1}{2}, 1) \cup (1, \infty]$, we define
  \begin{align*}
      H^{\uparrow}_\alpha(A|B)_{\rho} = -\inf_{\sigma_B} D_{\alpha}(\rho_{AB} \| \ident_A \otimes \sigma_B) 
  \end{align*}
  where the infimum is over all subnormalized density operators on $B$. Note that we also refer to $H_{\infty}^{\uparrow}(A|B)_{\rho}$ as $H_{\min}(A|B)_{\rho}$, called the \emph{min-entropy}, and to $H^{\uparrow}_{\frac{1}{2}}(A|B)_{\rho}$ as $H_{\max}(A|B)_{\rho}$, called the \emph{max-entropy}.
\end{definition}

Finally, in the case of the min- and max-entropy, we will also need ``smooth'' versions. These are versions of the min- and max-entropy where we compute the entropy for the best state within $\varepsilon$ of the actual state, where the distance is given by the purified distance. We begin by defining the purified distance~\cite{r02,r03,gln05,r06,TCR10,Tom12}:

\begin{definition}[Purified distance]
    Let $\rho$ and $\sigma$ be two subnormalized density operators. Then, the purified distance between $\rho$ and $\sigma$ is given by
    \[ P(\rho, \sigma) := \sqrt{1 - \left( \| \sqrt{\rho} \sqrt{\sigma} \|_1 + \sqrt{(1-\tr[\rho])(1-\tr[\sigma])} \right)^2}. \]
    \label{def:purified-distance}
\end{definition}
Note that this reduces to $P(\rho,\sigma) = \sqrt{1 - \| \sqrt{\rho} \sqrt{\sigma} \|_1^2}$ whenever either $\rho$ or $\sigma$ is normalized. We are now ready to define the smooth min- and max-entropy:
\begin{definition}\label{def:smooth-min-max-entropy}
  For any density operator $\rho_{A B}$ and for $\eps \in [0,1]$ the \emph{$\eps$-smooth min- and max-entropies of $A$ conditioned on $B$} are given by:
  \begin{align*}
      H_{\min}^\eps(A|B)_{\rho} & = \sup_{\tilde{\rho}: P(\rho,\tilde{\rho}) \leqslant\varepsilon} H_{\min}(A|B)_{\tilde{\rho}} \\
      H_{\max}^\eps(A|B)_{\rho} & = \inf_{\tilde{\rho}: P(\rho,\tilde{\rho}) \leqslant\varepsilon} H_{\max}(A|B)_{\tilde{\rho}}.
  \end{align*}
  respectively, where $\tilde{\rho}$ is any subnormalized density operator that is $\eps$-close to $\rho$ in terms of the purified distance~\cite{TCR10,Tom12}.
\end{definition}

\subsection{Chain rule for Rényi entropies} \label{sec:chain}
In~\cite{dfr16}, the central piece of the proof was a chain rule for Rényi entropies. As our proof largely follows the same steps, we reproduce the most relevant statement here for the reader's convenience. For the proofs, we refer the reader to~\cite{dfr16}. An important property of a tripartite state $\rho_{ABC}$ that we will be using throughout the paper is the Markov chain condition written $A \leftrightarrow B \leftrightarrow C$ and defined by $I(A : C | B)_{\rho} = 0$.

\begin{corollary}[Corollary 3.4 in \cite{dfr16}] 
\label{cor_conditionalmap}
Let $\rho^0_{R A_1 B_1}$ be a density operator on $R \otimes A_1 \otimes B_1$ and $\cM$ be a trace-preserving completely-positive map from $\mathcal{L}(R)$ to $\mathcal{L}(A_2 \otimes B_2)$. Assuming that $\rho_{A_1 B_1 A_2 B_2} = (\cM \otimes \cI_{A_1B_1})(\rho^0_{RA_1B_1})$ satisfies the Markov condition $A_1 \leftrightarrow B_1 \leftrightarrow B_2$, we have for $\alpha \in [\frac{1}{2},\infty)$
   \begin{align*} 
          \inf_{\omega} H_{\alpha}(A_2 | B_2 A_1 B_1)_{\cM(\omega)} 
  \leqslant H_\alpha(A_1 A_2 | B_1 B_2)_{\cM(\rho^0)} - H_{\alpha}(A_1 | B_1)_{\rho^0} \leqslant \sup_{\omega} H_{\alpha}(A_2 | B_2 A_1 B_1)_{\cM(\omega)} 
 \end{align*}
    where the supremum and infimum range over density operators $\omega_{R A_1 B_1}$ on $R \otimes A_1 \otimes B_1$. Moreover, if $\rho^0_{R A_1 B_1}$ is pure then we can optimise over pure states $\omega_{R A_1 B_1}$.
\end{corollary}

\section{The quantum divergence variance and its properties}\label{sec:divergence-variance}
The second-order term in our main result will be governed by a quantity called the quantum divergence variance, defined as follows:
\begin{definition}[Quantum divergence variance]
    Let $\rho, \sigma$ be positive semidefinite operators such that $D(\rho \| \sigma)$ is finite (i.e.~$\Supp(\rho) \subseteq \Supp(\sigma)$). Then, the quantum divergence variance $V(\rho \| \sigma)$ is defined as:
    \begin{align*}
        V(\rho \| \sigma) &:= \frac{1}{\tr[\rho]} \tr\left[ \rho(\log \rho - \log \sigma - \ident D(\rho \| \sigma))^2 \right]\\
        &= \frac{1}{\tr[\rho]} \tr\left[ \rho(\log \rho - \log \sigma)^2 \right] - D(\rho \| \sigma)^2.
    \end{align*}
\end{definition}
This was already defined in~\cite{th12} and~\cite{l14-3} under the names ``quantum information variance'' and ``quantum relative variance'' respectively; we instead choose a different name to clearly mark its relation to the divergence and to avoid confusion with the other variances that we are about to define.

\begin{definition}[Quantum conditional entropy variance]
    Let $\rho_{AB}$ be a bipartite quantum state. Then, the quantum conditional entropy variance $V(A|B)_{\rho}$ is given by:
    \begin{align*}
        V(A|B)_{\rho} &:= V(\rho_{AB} \| \ident_A \otimes \rho_B).
    \end{align*}
\end{definition}
    Likewise, this was already defined in~\cite{th12} under the name ``quantum conditional information variance''. Of course, the system in the conditioning can be omitted in the unconditional case. Finally, we define the quantum mutual information variance, first defined in~\cite{dtw14}:
\begin{definition}[Quantum mutual information variance]
    Let $\rho_{AB}$ be a bipartite quantum state. Then, the quantum mutual information variance $V(A;B)_{\rho}$ is given by:
    \begin{align*}
        V(A; B)_{\rho} &:= V(\rho_{AB} \| \rho_A \otimes \rho_B).
    \end{align*}
\end{definition}

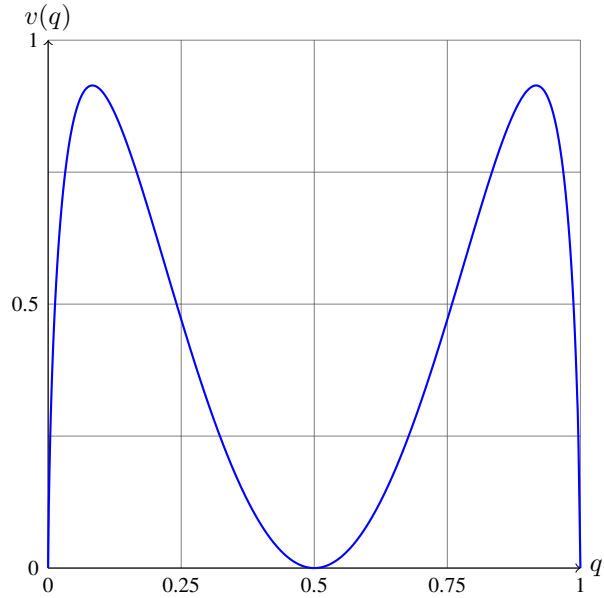
\begin{figure}
    \begin{center}
\begin{tikzpicture}[scale=7]
	\draw[very thin,color=gray,step=0.25] (0,0) grid (1,1);
	\draw[->] (0,0) -- (1,0) node[right] {$q$};
	\foreach \x in {0, 0.25, 0.5, 0.75, 1}
		\node[below,font=\footnotesize] at (\x cm,0) {\x};
		\draw[->] (0,0) -- (0,1) node[above] {$v(q)$};
	\foreach \y in {0,  0.5,  1}
		\node[left,font=\footnotesize] at (0,\y cm) {\y};
    \draw[blue, thick] (0,0) plot[smooth] file{entropy-variance-of-bernoulli.dat};
\end{tikzpicture}
\end{center}
\caption{Plot of $v(q) = V(X)$, where $X$ is a Bernoulli RV with $\Pr[X = 0] = q$. It peaks at around $v(0.083) \approx 0.9142$.}
\label{fig:divergence-variance-of-bernoulli}
\end{figure}

These various quantities have a number of elementary properties that we prove here. First, to get a sense of what the divergence variance looks like in a simple case, we plot the divergence variance of a single bit $X$ with $\Pr[X=0]$ in Figure~\ref{fig:divergence-variance-of-bernoulli}. We also note that the divergence variance does not satisfy the data processing inequality, even in the classical case; in other words, it is not true in general that $V(\rho \| \sigma) \geqslant V(\mathcal{E}(\rho) \| \mathcal{E}(\sigma))$ for a quantum channel $\mathcal{E}$. To see this, consider the following counterexample: let $\rho = \proj{0}$, $\sigma = \ident$, and let $\mathcal{E}$ be a binary symmetric channel in the computational basis with error rate $0.083$. Then, we can see from the plot in Figure~\ref{fig:divergence-variance-of-bernoulli} that $V(\mathcal{E}(\rho) \| \mathcal{E}(\sigma)) > V(\rho \| \sigma)$. It is also easy to see that the opposite inequality is also false in general.

Now, we show that the divergence variance obeys the following basic bounds:
\begin{lemma}[General bounds]\label{lem:divergence-variance-general-bounds}
    For any positive semidefinite operators $\rho, \sigma$, with $\Supp(\rho) \subseteq \Supp(\sigma)$,
and any $\nu \in (0,1)$,
\begin{align}
\label{eq:bound_dalpha}
V(\rho \| \sigma) \leq \frac{1}{\nu^2} \log^2\left(2^{-\nu D(\rho \| \sigma) + \nu D'_{1+\nu}(\rho \| \sigma)} + 2^{\nu D(\rho \| \sigma) - \nu D'_{1-\nu}(\rho \| \sigma)}+1\right).
\end{align}
 \end{lemma}
\begin{proof}
    First, without loss of generality, we restrict the space to the support of $\sigma$. We then proceed in a way similar to~\cite[Lemma 8]{TCR09}. We introduce  $X = 2^{-D(\rho \| \sigma)} \rho  \otimes (\sigma^{-1})^T$, $\ket{\varphi} = (\sqrt{\rho} \otimes \ident) \ket{\gamma}$ with $\ket{\gamma} = \sum_{i} \ket{i} \otimes \ket{i}$. We then have $V(\rho \| \sigma) = \frac{1}{\ln^2 2} \bra{\varphi} \ln^2 X \ket{\varphi}$. Observe that we have for $\nu \in (0,1)$ and any $t > 0$
   \begin{align*}
   \ln^2 t &= \frac{1}{\nu^2} \ln^2 t^{\nu} \\
   &\leq \frac{1}{\nu^2} \left(\ln\left(t^{\nu}+\frac{1}{t^{\nu}} \right) \right)^2 \\
   &\leq \frac{1}{\nu^2} \left(\ln\left(t^{\nu}+\frac{1}{t^{\nu}} + 1 \right) \right)^2 \ ,
   \end{align*}
   where in the first inequality, we used the fact that $\ln(x)^2 \leq \ln(x+\frac{1}{x})^2$ for any $x > 0$ and in the second inequality the fact that $x+\frac{1}{x} \geqslant1$.  As a result, we have that
   \[ (\Pi_{\Supp(\rho)} \otimes \ident) \ln^2 X (\Pi_{\Supp(\rho)} \otimes \ident) \leqslant \frac{1}{\nu^2} (\Pi_{\Supp(\rho)} \otimes \ident)\left(\ln \left(X^{\nu} + \frac{\ident}{X^{\nu}} + \ident \right) \right)^2(\Pi_{\Supp(\rho)} \otimes \ident) \]
   and therefore
   \[
       \bra{\varphi} \ln^2 X \ket{\varphi} \leq \frac{1}{\nu^2} \bra{\varphi} \left(\ln \left(X^{\nu} + \frac{\ident}{X^{\nu}} + \ident \right) \right)^2 \ket{\varphi},
   \]
   We now use the fact that the function $s \mapsto \ln^2(s)$ is concave on the interval $[e, +\infty)$ and that $\ket{\varphi}$ is in the span of the eigenvectors of $X^{\nu} + \frac{\ident}{X^{\nu}} + \ident$ with eigenvalues in $[3,\infty)$) to get
    \[
   \bra{\varphi} \ln^2 X \ket{\varphi} \leq \frac{1}{\nu^2} \ln^2 \left( \bra{\varphi} X^{\nu} \ket{\varphi} + \bra{\varphi} \frac{\ident}{X^{\nu}} \ket{\varphi} + 1\right) . \]
   But observe that 
   \[
   \bra{\varphi} X^{\nu} \ket{\varphi} = 2^{-\nu D(\rho \| \sigma)} \tr(\rho^{1+\nu} \sigma^{-\nu}) = 2^{-\nu D(\rho \| \sigma) + \nu D'_{1+\nu}(\rho \| \sigma)}
   \]
   and 
   \[
   \bra{\varphi} \frac{\ident}{X^{\nu}} \ket{\varphi} = 2^{+\nu D(\rho \| \sigma)} \tr(\rho^{1-\nu} \sigma^{\nu}) = 2^{+\nu D(\rho \| \sigma) - \nu D'_{1-\nu}(\rho \| \sigma)} \ .
   \]
\end{proof}

This leads to the following bounds for the conditional entropy variance and the mutual information variance:
\begin{corollary} \label{lem_var_dim} 
For any density operator $\rho_{AB}$, we have
\begin{align*}
V(A|B)_{\rho} &\leq \log^2(2d_{A}^2+1) \\
V(A;B)_{\rho} &\leq 4\log^2(2d_A + 1) \ .
\end{align*}
Moreover, if the system $A$ is classical, then the upper bounds can be improved to 
\begin{align*}
V(A|B)_{\rho} &\leq \log^2(2d_{A}+1) \\
V(A;B)_{\rho} &\leq 4\log^2(2\sqrt{d_A} + 1) \ .
\end{align*}
\end{corollary}
\begin{proof}
For the upper bound on $V(A|B)_{\rho} = V(\rho_{AB} \| \id_{A} \otimes \rho_B)$, using~\eqref{eq:bound_dalpha} for $\nu \in (0,1)$, we get:
\begin{align*}
V(A|B)_{\rho}
&\leq \frac{1}{\nu^2}\log^2\left( 2^{-\nu D(\rho_{AB} \| \id_{A} \otimes \rho_{B}) + \nu D'_{1+\nu}(\rho_{AB} \| \id_{A} \otimes \rho_{B}) } + 2^{\nu D(\rho_{AB} \| \id_{A} \otimes \rho_{B}) - \nu D'_{1-\nu}(\rho_{AB} \| \id_{A} \otimes \rho_{B}) } + 1\right) \\
&= \frac{1}{\nu^2}\log^2\left( 2^{\nu(H(A|B)_{\rho} - H'_{1+\nu}(A|B)_{\rho}) } + 2^{\nu(-H(A|B)_{\rho} + H'_{1-\nu}(A|B)_{\rho}) } + 1\right) \\
&\leq \frac{1}{\nu^2}\log^2\left( 2d_{A}^{2\nu} + 1\right),
\end{align*}
where the first inequality uses Lemma~\ref{lem:divergence-variance-general-bounds} and the last inequality uses the fact that all the entropy terms are bounded by $-\log d_A \leqslant H_{\star}(A|\star) \leqslant \log d_A$ (see e.g., \cite[Lemma 5.2]{Tom15book}). Taking the limit $\nu \to 1$, we get the desired result. In the case where $\rho_{AB}$ is separable, we have instead $0 \leqslant H_{\star}(A|\star) \leqslant \log d_A$ which leads to the improved bound.


For the bound on $V(A;B) = V(\rho_{AB} \| \rho_{A} \otimes \rho_{B})$, we use~\eqref{eq:bound_dalpha} with $\nu = \frac{1}{2}$ to have an upper bound of the form
\begin{align*}
V(A;B)_{\rho}
&\leq 4 \log^2\left(2^{-\frac{1}{2} D(\rho_{AB} \| \rho_{A} \otimes \rho_B) + \frac{1}{2} D'_{\frac{3}{2}}(\rho_{AB} \| \rho_{A} \otimes \rho_{B})} + 2^{\frac{1}{2} D(\rho_{AB} \| \rho_{A} \otimes \rho_B) - \frac{1}{2} D'_{\frac{1}{2}}(\rho_{AB} \| \rho_{A} \otimes \rho_{B})} + 1\right) \\
&\leq 4 \log^2\left(2^{\frac{1}{2} D'_{\frac{3}{2}}(\rho_{AB} \| \rho_{A} \otimes \rho_{B})} + d_{A} + 1 \right) ,
\end{align*} 
where we used the fact that $D(\rho_{AB} \| \rho_{A} \otimes \rho_B)$ and $D'_{\frac{1}{2}}(\rho_{AB} \| \rho_{A} \otimes \rho_{B})$ are nonnegative and $D(\rho_{AB} \| \rho_{A} \otimes \rho_B) \leq 2 \log d_A$.
To conclude, it suffices to show that $D'_{\frac{3}{2}}(\rho_{AB} \| \rho_{A} \otimes \rho_{B}) \leq 2 \log d_{A}$. To do this, let $\rho_{ABC}$ be a purification of $\rho$. We then have that: 
\begin{align*}
    D'_{\frac{3}{2}}(\rho_{AB} \| \rho_A \otimes \rho_B) &\leqslant D'_{\frac{3}{2}}(\rho_{ABC} \| \rho_A \otimes \rho_{BC})\\
    &= 2 \log \tr\left[ \rho_{ABC}^{\frac{3}{2}} (\rho_A \otimes \rho_{BC})^{-\frac{1}{2}} \right]\\
    &\leqslant 2 \log \sqrt{\tr\left[\rho_{ABC}^{\frac{3}{2}} \rho_A^{-1} \right] \tr\left[ \rho_{ABC}^{\frac{3}{2}} \rho_{BC}^{-1} \right]}\\
    &= \log \left[ \tr\left[\rho_{ABC} \rho_A^{-1} \right] \tr\left[ \rho_{ABC} \rho_{BC}^{-1} \right] \right]\\
    &\leqslant \log d_A + \log \dim \Supp(\rho_{BC})\\
    &\leqslant 2\log d_A.
\end{align*}
We remark that the choice of looking at $D'_{\frac{3}{2}}$ is not arbitrary. In fact, $D'_{1+\nu}(\rho_{AB} \| \rho_{A} \otimes \rho_B)$ for $\nu > \frac{1}{2}$ may be arbitrarily large as can be seen with the following example. Let $\ket{\Phi(\lambda)}_{AB} = \sqrt{\lambda} \ket{00}_{AB} + \sqrt{1-\lambda} \ket{11}_{AB}$ for $\lambda \in [0,1]$. We set $\rho_{AB} = \proj{\Phi(\lambda)}_{AB}$. Then, we can compute
\begin{align*}
D'_{1+\nu}(\rho_{AB} \| \rho_{A} \otimes \rho_B)
&= \frac{1}{\nu} \log \tr\left[\lambda^{1-2\nu} \proj{00} + (1-\lambda)^{1-2\nu} \proj{11} \right],
\end{align*}
which diverges as $\lambda \to 0$ for $\nu > \frac{1}{2}$.

When the system $A$ is classical, then we have $D'_{\frac{3}{2}}(\rho_{AB} \| \rho_{A} \otimes \rho_{B}) \leq \log d_{A}$. 
In fact, we write $\rho_{AB} = \sum_{a} p(a) \proj{a}_{A} \otimes \rho_{B}(a)$, where $\{p(a)\}_a$ is a probability distribution and $\rho_{B}(a)$ are density operators. Then, we compute
\begin{align*}
    D'_{\frac{3}{2}}(\rho_{AB} \| \rho_A \otimes \rho_B) 
        &= 2 \log \sum_{a} \tr\left[ \left(p(a)^{\frac{3}{2}} \proj{a} \otimes \rho_{B}(a)^{\frac{3}{2}} \right) (p(a)^{-\frac{1}{2}} \proj{a} \otimes \rho_{B}^{-\frac{1}{2}}) \right] \\
        &= 2 \log \sum_{a} \tr\left[ p(a) \rho_{B}(a)^{\frac{3}{2}} \rho_{B}^{-\frac{1}{2}} \right].
%
 \end{align*}
  Now note that for any $a$, $\rho_B \geq p(a) \rho_{B}(a)$, and thus by operator monotonicity of $x \mapsto -x^{-\frac{1}{2}}$, we have $\rho_{B}^{-\frac{1}{2}} \leq p(a)^{-\frac{1}{2}} \rho_B(a)^{-\frac{1}{2}}$, we get
  \begin{align*}
    D'_{\frac{3}{2}}(\rho_{AB} \| \rho_A \otimes \rho_B) 
        &\leqslant 2 \log \sum_{a} p(a)^{\frac{1}{2}} \\
        &\leqslant \log d_{A}.
%
 \end{align*}

\end{proof}

Next, we show that the divergence variance is additive, in the following sense:

\begin{lemma}[Additivity of the divergence variance]\label{lem:divergence-variance-additivity}
    Let $\rho, \tau$ be density operators and $\sigma, \omega$ be positive semidefinite operators. Then,
    \[ V(\rho \otimes \tau \| \sigma \otimes \omega) = V(\rho \| \sigma) + V(\tau \| \omega). \]
\end{lemma}
\begin{proof}
    We have that
    \begin{multline*}
        V(\rho \otimes \tau \| \sigma \otimes \omega)\\
        \begin{aligned}
        &= \tr\left[(\rho \otimes \tau)\left(\log (\rho \otimes \tau) - \log(\sigma \otimes \omega) - \ident D(\rho \| \sigma) - \ident D(\tau \| \omega) \right)^2 \right]\\
        &= \tr\left[(\rho \otimes \tau)\left(\log \rho \otimes \ident + \ident \otimes \log \tau - \log \sigma \otimes \ident - \ident \otimes \log \omega - \ident D(\rho \| \sigma) - \ident D(\tau \| \omega) \right)^2 \right]\\
        &= \tr\left[ (\rho \otimes \tau)\left( \log \rho \otimes \ident - \log \sigma \otimes \ident - \ident D(\rho \| \sigma) \right)^2 \right]\\
        &\quad\quad\quad\quad + \tr\left[ (\rho \otimes \tau)\left( \ident \otimes \log \tau - \ident \otimes \log \omega - \ident D(\tau \| \omega) \right)^2 \right]\\
        &\quad\quad\quad\quad + \tr\left[ (\rho \otimes \tau)\left( \log \rho \otimes \ident - \log \sigma \otimes \ident - \ident D(\rho \| \sigma) \right) \left( \ident \otimes \log \tau - \ident \otimes \log \omega - \ident D(\tau \| \omega) \right)\right]\\
        &\quad\quad\quad\quad + \tr\left[ (\rho \otimes \tau) \left( \ident \otimes \log \tau - \ident \otimes \log \omega - \ident D(\tau \| \omega) \right)\left( \log \rho \otimes \ident - \log \sigma \otimes \ident - \ident D(\rho \| \sigma) \right)\right]\\
        &= V(\rho \| \sigma) + V(\tau \| \omega)\\
        &\quad\quad\quad\quad + 2\tr\left[ \rho (\log \rho - \log \sigma - \ident D(\rho \| \sigma)) \right] \tr\left[\tau (\log \tau - \log \omega - \ident D(\tau \| \omega)) \right]\\
        &= V(\rho \| \sigma) + V(\tau \| \omega).
        \end{aligned}
    \end{multline*}
\end{proof}

We also show that a conditional entropy variance with a classical variable $X$ in the conditioning admits a decomposition in terms of the possible values of $X$:
\begin{lemma}\label{lem:divergence-variance-classical-X}
    Let $\rho_{ABX}$ be a tripartite state with $X$ classical. Then,
    \begin{align*}
    V(A|BX)_{\rho} &= \sum_x p_x V(A|B,X=x) + \Variance(W)
    \end{align*}
    where $W$ is a random variable that takes value $H(A|B, X=x)$ with probability $p_x$. In particular,
    \[ V(A|BX)_{\rho} \geqslant \sum_x p_x V(A|B,X=x). \]
\end{lemma}
\begin{proof}
    We have that
    \begin{align*}
        V(A|BX)_{\rho} &= \tr\left[ \rho_{ABX} \left( \log \rho_{ABX} - \ident_A \otimes \log \rho_{BX} \right)^2 \right] - H(A|BX)^2\\
        &= \sum_x p_x \tr\left[ \rho_{AB|X=x} \left( \log \rho_{AB|X=x} + \ident \log p_x - \ident_A \otimes \log \rho_{B|X=x} - \ident \log p_x \right)^2 \right]\\
        &\quad\quad\quad\quad\quad - H(A|BX)^2\\
        &= \sum_x p_x \tr\left[ \rho_{AB|X=x} \left( \log \rho_{AB|X=x} - \ident_A \otimes \log \rho_{B|X=x} \right)^2 \right]\\
        &\quad\quad\quad\quad\quad - \left( \sum_x p_x H(A|B,X=x) \right)^2\\
        &= \sum_x  p_x \left( V(A|B,X=x) + H(A|B,X=x)^2 \right) - \left( \sum_x p_x H(A|B,X=x) \right)^2\\
        &= \sum_x p_x V(A|B,X=x) - \left( \sum_x p_x H(A|B,X=x) \right)^2 + \sum_x p_x H(A|B,X=x)^2\\
        &= \sum_x p_x V(A|B,X=x) - \left( \mbE W \right)^2 + \mbE[W^2]\\
        &= \sum_x p_x V(A|B,X=x) + \Variance(W).
    \end{align*}
\end{proof}

We will also need the following decomposition of the conditional entropy variance for Markov chains:
\begin{lemma}
\label{lem_var_markov}
    Let $\rho_{ABCDX}$ be a quantum state with $X$ classical satisfying the Markov chain $AC \leftrightarrow X \leftrightarrow BD$; i.e.~$I(AC:BD|X)=0$. Then,
    \[ V(AB|CDX) = V(A|CX) + V(B|DX) + 2\Cov(W_1, W_2), \]
    where $W_1$ and $W_2$ are random variables that take value $H(A|C,X=x)$ and $H(B|D,X=x)$ according to the value of $X$, respectively. In particular, this shows that for a trivial $B$ system, $V(A|CDX) = V(A|CX)$.
\end{lemma}    
\begin{proof}
    We perform the computation as follows:
    \begin{multline*}
        V(AB|CDX)\\
        \begin{aligned}
        &= \sum_x p_x V(AB|CD,X=x) + \Variance(W_1 + W_2)\\
        &= \sum_x p_x ( V(A|C,X=x) + V(B|D, X=x) ) + \Variance(W_1) + \Variance(W_2) + 2\Cov(W_1,W_2)\\
        &= V(A|CX) + V(B|DX) + 2\Cov(W_1,W_2),
        \end{aligned}
    \end{multline*}
    where the first equality follows from Lemma~\ref{lem:divergence-variance-classical-X}, and the second equality from Lemma~\ref{lem:divergence-variance-additivity}.
\end{proof}
    
Finally, the following more specialized lemmas will be needed in the proof of our main result:
\begin{lemma}\label{lem_var_decomp}
Let $\rho_{ACD\bar{D}X}$ be a quantum state with $X$ classical that can be written as
\[ \sum_x p_x \proj{x}_{X} \otimes \rho^{(x)}_{AC} \otimes \tau^{(x)}_{D\bar{D}} \]
with $\tau^{(x)}_{\bar{D}} = \frac{\ident_{\bar{D}}}{d_{\bar{D}}}$ for all $x$. Then,
    \[ V(ADX|C\bar{D}) \leqslant V(AX|C) + V(D|X\bar{D}) + 2 \sqrt{V(AX|C) V(D|X\bar{D})}. \]
    \end{lemma}
    \begin{proof}
        First, note that the chain rule together with the form of the state in the lemma gives $H(ADX|C\bar{D}) = H(AX|C) + H(D|X\bar{D})$. We can then proceed as follows:
        \begin{align*}
            V(ADX|C\bar{D})
            &= \sum_{x} p_{x} \tr\Big[ \rho^{(x)}_{AC} \otimes \tau^{(x)}_{D\bar{D}} \Big( \log p_x \rho^{(x)}_{AC} \otimes \ident_{D\bar{D}} \\
            &\qquad + \ident_{AC} \otimes \log \tau^{(x)}_{D\bar{D}} - \ident_{AD\bar{D}} \otimes \log \rho_{C} + \ident_{ACD\bar{D}} \log d_{\bar{D}}  + \ident H(ADX|C\bar{D})\Big)^2 \Big] \\
            &= \sum_{x} p_{x} \tr\Big[ \rho^{(x)}_{AC} \otimes \tau^{(x)}_{D \bar{D}} \Big( \log p_{x}\rho^{(x)}_{AC} \otimes \ident_{D\bar{D}} - \ident_{AD\bar{D}} \otimes \log \rho_{C} + \ident H(AX|C)\Big)^2 \Big] \\
            &+ \sum_{x} p_{x} \tr\Big[ \rho^{(x)}_{AC} \otimes \tau^{(x)}_{D\bar{D}} \Big( \ident_{AC} \otimes \log \tau^{(x)}_{D\bar{D}}  + \ident_{ACD\bar{D}} \log d_{\bar{D}} + \ident H(D|X\bar{D})\Big)^2 \Big] \\
            &+ 2\sum_{x} p_{x} \tr\Big[ \rho^{(x)}_{AC} \otimes \tau^{(x)}_{D\bar{D}} \Big( \log p_{x}\rho^{(x)}_{AC} \otimes \ident_{D\bar{D}} - \ident_{AD\bar{D}} \otimes \log \rho_{C} + \ident H(AX|C)\Big) \\
            &\qquad \qquad \Big( \ident_{AC} \otimes \log \tau^{(x)}_{D\bar{D}}  + \ident_{ABD\bar{D}} \log d_{\bar{D}} + \ident H(D|X\bar{D})\Big) \Big] \\
            &= V(AX|C) + V(D|X\bar{D}) + 2 \cdot \text{crossterm}, 
         \end{align*}
         where we used in the second equality the fact that $\Big( \log p_{x}\rho^{(x)}_{AC} \otimes \ident_{D\bar{D}} - \ident_{AD\bar{D}} \otimes \log \rho_{C} + \ident H(AX|C)\Big)$ and $\Big( \log \tau^{(x)}_{D\bar{D}} \otimes \ident_{AC} + \ident_{ABD\bar{D}} \log d_{\bar{D}} + \ident H(D|X\bar{D})\Big)$ commute. To get the last equality, we observe that
         \begin{align*}
         &\sum_{x} p_{x} \tr\Big[ \rho^{(x)}_{AC} \otimes \tau^{(x)}_{D \bar{D}} \Big( \log p_{x}\rho^{(x)}_{AC} \otimes \ident_{D\bar{D}} - \ident_{AD\bar{D}} \otimes \log \rho_{C} + \ident H(AX|C)\Big)^2 \Big] \\
         &= \sum_{x} p_x \tr\Big[ \rho^{(x)}_{AC} \Big( \log p_{x}\rho^{(x)}_{AC}  - \ident_{A} \otimes \log \rho_{C} + \ident H(AX|C)\Big)^2 \Big] \\
	&= V(AX | C),         
         \end{align*}
         and 
         \begin{align*}
         &\sum_{x} p_{x} \tr\Big[ \rho^{(x)}_{AC} \otimes \tau^{(x)}_{D\bar{D}} \Big( \log \tau^{(x)}_{D\bar{D}} \otimes \ident_{AC} + \ident_{ACD\bar{D}} \log d_{\bar{D}} + \ident H(D|X\bar{D})\Big)^2 \Big] \\
         &= \tr\Bigg[ \sum_{x} p_{x} \proj{x} \otimes \tau^{(x)}_{D\bar{D}} \Bigg( \log\left(\sum_{x} p_{x} \proj{x} \otimes \tau^{(x)}_{D\bar{D}}\right) - \ident_{D} \otimes \log \left( \sum_{x} p_{x} \proj{x} \otimes \tau^{(x)}_{\bar{D}} \right)\\
         &\quad\quad\quad\quad\quad + \ident_{X D \bar{D}} H(D|X\bar{D})\Bigg)^2 \Bigg] \\
         	&= V(D | \bar{D} X),         
         \end{align*}
         We are now going to bound the cross term by applying the Cauchy-Schwarz inequality. Using the cyclicity of the trace for $(\rho^{(x)}_{AC})^{1/2}$, we have
         \begin{align*}
         \text{crossterm} &= \tr \Bigg[ \left( \sum_{x} \sqrt{p_{x}} \proj{x} \otimes (\rho^{(x)}_{AC})^{1/2} \Big( \log p_{x}\rho^{(x)}_{AC} - \ident_{A} \otimes \log \rho_{C} + \ident H(AX|C)\Big) \otimes (\tau^{(x)}_{D\bar{D}})^{1/2} \right) \\
         & \qquad \cdot \left( \sum_{x} \sqrt{p_{x}} \proj{x}  \otimes (\rho^{(x)}_{AC})^{1/2} \otimes (\tau^{(x)}_{D\bar{D}})^{1/2} \Big( \log \tau^{(x)}_{D\bar{D}} + \ident_{D\bar{D}} \log d_{\bar{D}} + \ident H(D|X\bar{D})\Big)  \right)  \Bigg] \\
       	&\leqslant \sqrt{\tr(Y Y^{\dagger}) \tr(Z Z^{\dagger})} ,
         \end{align*}
         where $Y = \sum_{x} \sqrt{p_{x}} \proj{x} \otimes (\rho^{(x)}_{AC})^{1/2} \Big( \log p_{x}\rho^{(x)}_{AC} - \ident_{A} \otimes \log \rho_{C} + \ident H(AX|C)\Big) \otimes (\tau^{(x)}_{D\bar{D}})^{1/2}$ and $Z = \sum_{x} \sqrt{p_{x}} \proj{x}  \otimes (\rho^{(x)}_{A\bar{C}})^{1/2} \otimes (\tau^{(x)}_{D\bar{D}})^{1/2} \Big( \log \tau^{(x)}_{D\bar{D}} + \ident_{D\bar{D}} \log d_{\bar{D}} + \ident H(D|X\bar{D})\Big)$. We conclude by observing that $\tr(YY^{\dagger}) = V(AX|C)$ and $\tr(ZZ^{\dagger}) = V(D|X\bar{D})$.
    \end{proof}

\begin{lemma}
\label{lem_chain_rule_var}
For any state $\rho_{ABC}$, we have
   \begin{align}
        V(AC|B)_{\rho} &= V(A|B)_{\rho} + V(C|BA)_{\rho} \notag \\
        &+ \tr\left( \rho_{ABC} (\log \rho_{AB} - \log \rho_{B} + H(A|B))(\log \rho_{ABC} - \log \rho_{AB} + H(C|BA) ) \right) \notag \\
        &+ \tr\left( \rho_{ABC} (\log \rho_{ABC} - \log \rho_{AB} + H(C|BA) ) (\log \rho_{AB} - \log \rho_{B} + H(A|B)) \right).\label{eq_chain_rule_var} 
    \end{align}
\end{lemma}
\begin{proof}
    Direct calculation.
\end{proof}

\begin{lemma}
\label{lem_det_func_var}
Let $\rho_{XAB}$ be of the form $\rho_{XAB} = \sum_{x \in \cX} \proj{x}_{X} \otimes \rho_{AB, x}$ with $\tr(\rho_{AB, x} \rho_{AB, x'}) = 0$ when $x \neq x'$.
  Then we have
   \begin{align}
        V(AX|B)_{\rho} &= V(A|B)_{\rho} \ . \label{eq_det_func_var}
    \end{align}
\end{lemma}
In other words, if the states $\rho_{AB,x}$ are orthogonal for different values of $x$, then this effectively makes the subsystem $X$ redundant for the purpose of computing the conditional entropy variance.
\begin{proof}
Using Lemma~\ref{lem_chain_rule_var}, it suffices to show only the first term of~\eqref{eq_chain_rule_var} remains. In fact, we have $H(X|BA) = 0$ and
\begin{align*}
&V(X|BA)_{\rho} = \tr\left(\rho_{XAB} (\log \rho_{XAB} - \log \rho_{AB})^2 \right) \\
&= \sum_{x} \tr\left( \proj{x}_{X} \otimes \rho_{AB, x} \left(\sum_{x'} \proj{x'}_{X} \otimes \log \rho_{AB, x'} - \ident_{X} \otimes \sum_{x'} \log \rho_{AB,x'}\right)^2 \right) \\
&= \sum_{x} \tr\left( \proj{x}_{X} \otimes \rho_{AB, x} \left( \sum_{x'} (\proj{x'}_{X} - \ident_{X})^2 \otimes \log^2 \rho_{AB, x'} \right) \right) \\
&= \sum_{x} \tr\left( \proj{x}_{X} (\proj{x}_{X} - \ident_{X})^2 \otimes \rho_{AB, x}  \log^2 \rho_{AB, x} \right) \\
&= 0 \ .
\end{align*}
In addition, the other terms are also zero:
\begin{align*}
&\tr\left(\rho_{XAB} (\log \rho_{AB} - \log \rho_{B} + H(A|B))(\log \rho_{XAB} - \log \rho_{AB})\right) \\
&= \sum_{x} \tr\Bigg( (\proj{x} \otimes \rho_{AB,x}) (\ident_{X} \otimes (\log \rho_{AB} - \log \rho_{B} + \ident_{AB} H(A|B))) (\proj{x} \otimes \log \rho_{AB,x}\\
&\quad\quad\quad\quad- \ident_{X} \otimes \sum_{x'} \log \rho_{AB,x'}) \Bigg) \\
&= \sum_{x} \tr\left( \rho_{AB,x} (\log \rho_{AB} - \log \rho_{B} + \ident_{AB} H(A|B)) (\sum_{x' \neq x} \log \rho_{AB,x'}) \right) \\
&= 0 \ ,
\end{align*}
using the orthogonality of $\rho_{AB,x}$ and $\rho_{AB,x'}$, and
\begin{multline*}
    \tr\left[ \rho_{XAB} (\log \rho_{XAB} - \log \rho_{AB})(\log \rho_{AB} - \log \rho_B + H(A|B))   \right]\\
    \begin{aligned}
        &= \sum_x \tr\left[ (\proj{x} \otimes \rho_{AB,x}) (\proj{x} \otimes \rho_{AB,x} - \ident_X \otimes \rho_{AB}) (\log \rho_{AB} - \log \rho_B + \ident_{XAB} H(A|B)) \right]\\
        &= \sum_x \tr\left[ (\proj{x} \otimes \rho_{AB,x}) (\proj{x} \otimes \rho_{AB,x} - \ident_X \otimes \rho_{AB,x}) (\log \rho_{AB} - \log \rho_B + \ident_{XAB} H(A|B)) \right]\\
        &= 0,
    \end{aligned}
\end{multline*}
where we have used the fact that $\rho_{AB,x} \rho_{AB} = \rho_{AB,x}^2$ by the orthogonality conditions.
\end{proof}

\section{Continuity bounds for Rényi divergences}\label{sec:renyi-vs-von-neumann}

A critical step in the proof is an explicit continuity bound for $D_{\alpha}$ when $\alpha$ approaches $1$. One such bound is given~\cite[Section 4.2.2]{Tom15book}. However, this bound does not give explicit values for the remainder term. The following lemma computes an explicit remainder term for the case of classical probability distributions. As in~\cite[Section 4.4.2]{Tom15book}, we will then apply this lemma to Nussbaum-Szko\l{}a distributions to get a similar result for the Petz divergence $D'_{\alpha}$ between quantum states.

\begin{lemma}
\label{lem_HalphaH_second_order_new}
Let $\rho$ be a density operator and $\sigma$ be a not necessarily normalized positive semidefinite operator. Let $\alpha > 1$ and $\mu \in (0,1)$. Then, we have that
\[ D_{\alpha}(\rho \| \sigma) \leqslant D'_{\alpha}(\rho \| \sigma) \leqslant D(\rho \| \sigma) + \frac{(\alpha-1) \ln 2}{2} V(\rho \| \sigma) + (\alpha-1)^2 K_{\rho,\sigma}, \]
    where
    \[ K_{\rho,\sigma}(\alpha,\mu) =  \frac{1}{6 \mu^3 \ln 2} 2^{(\alpha-1)(D'_{\alpha}(\rho \| \sigma) - D(\rho \| \sigma))} \ln^3\left( 2^{(\alpha + \mu -1)(D'_{\alpha+\mu}(\rho \| \sigma) - D(\rho \| \sigma))} + e^2 \right). \]
\end{lemma}
\begin{proof}
    As mentioned above, we start by proving the statement for classical probability distributions. For this proof, it will be more convenient for us to do everything using natural logarithms; we will therefore use the ``hatted'' quantities $\hat{D}$ and $\hat{V}$ for all relative entropies and variances to denote their counterparts defined using the natural logarithm. Let $P$ be a probability distribution and $Q$ be a not necessarily normalized distribution. Define the random variable $X$ with distribution $P$, and let $Z = e^{-\hat{D}(P \| Q)} \frac{P(X)}{Q(X)}$. Note that for any $\nu > 0$, we have
\begin{align}
\mbE[Z^\nu] 
&= e^{-\nu \hat{D}(P \| Q)} \sum_{x} P(x)^{1+\nu} Q(x)^{-\nu} \\
&= e^{-\nu (\hat{D}(P \| Q) - \hat{D}_{1+\nu}(P \| Q))}.
\end{align}
Now, letting $\nu = \alpha - 1$, we have
\begin{align}
    \hat{D}_{\alpha}(P \| Q) &= \frac{1}{\nu} \ln\left( \mbE[Z^\nu]\right) + \hat{D}(P \| Q).
    \end{align}
Applying Taylor's inequality to the function $\nu \mapsto \mbE[Z^{\nu}]$ we have
\begin{align}
\mbE[Z^{\nu}] \leq 1 + \nu \mbE[\ln Z] + \frac{\nu^2}{2} \mbE[\ln^2 Z] + \frac{\nu^3}{6} \sup_{0 < \gamma \leqslant \nu} \mbE[Z^{\gamma} \ln^3 Z].
\end{align}
Using the fact that $\mbE[\ln Z] = 0$ and
\begin{align}
    \mbE[\ln^2 Z] &= \sum_{x} P(x) \left(\ln \frac{P(x)}{Q(x)} - \hat{D}(P \| Q) \right)^2 \\
    &= \hat{V}(P \| Q) ,
\end{align}
together with the inequality $\ln(1+x) \leqslant x$, we get
\begin{align}
    \hat{D}_{\alpha}(P \| Q) &\leqslant \hat{D}(P \| Q) + \frac{\nu}{2} \hat{V}(P \| Q) + \frac{\nu^2}{6} \sup_{0 < \gamma \leqslant \nu} \mbE[Z^{\gamma} \ln^3 Z].
\end{align}
We now need to bound the remainder term. We want to use the concavity of $\ln^3$, but it is only concave on $[e^2, \infty)$. Hence, we start by using the fact that $\ln^3$ is nondecreasing and $Z^{\gamma} \geqslant 0$ to get
\begin{align}
    \mbE[Z^{\gamma} \ln^3 Z] &= \frac{1}{\mu^3} \mbE[Z^{\gamma} \ln^3(Z^{\mu})]\\
    &\leqslant \frac{1}{\mu^3} \mbE[Z^{\gamma} \ln^3 (Z^{\mu} + e^2)]\\
    &= \frac{1}{\mu^3} \mbE[Z^{\gamma}] \frac{\mbE[Z^{\gamma} \ln^3 (Z^{\mu} + e^2)]}{\mbE[Z^{\gamma}]},
    \end{align}
    for any $\mu \in (0,1]$. Then we use the concavity of the function $t \mapsto \ln^3(t+e^2)$ on $[0,\infty)$ and get
    \begin{align}
        \mbE[Z^{\gamma} \ln^3 Z]  &\leqslant \frac{1}{\mu^3} \mbE[Z^{\gamma}] \ln^3\left( \frac{\mbE[Z^{\gamma} (Z^{\mu} + e^2)]}{\mbE[Z^{\gamma}]} \right)\\
        &= \frac{1}{\mu^3} \mbE[Z^{\gamma}] \ln^3\left( \frac{\mbE[Z^{\mu+\gamma}]}{\mbE[Z^{\gamma}]} + e^2 \right)\\
        &= \frac{1}{\mu^3} e^{\gamma ( \hat{D}_{1+\gamma}(P \| Q) - \hat{D}(P \| Q))} \ln^3\left( \frac{e^{(\mu+\gamma)(\hat{D}_{1+\mu+\gamma}(P \| Q) - \hat{D}(P\|Q))}}{e^{\gamma (\hat{D}_{1+\gamma}(P \| Q) - \hat{D}(P\|Q))}} + e^2 \right) \\
        &\leqslant \frac{1}{\mu^3} e^{\gamma ( \hat{D}_{1+\gamma}(P \| Q) - \hat{D}(P \| Q))} \ln^3\left(e^{(\mu+\gamma)(\hat{D}_{1+\mu+\gamma}(P \| Q) - \hat{D}(P\|Q))} + e^2 \right) \ ,
\end{align}
where we used the fact that $\hat{D}_{1+\gamma}(P \| Q) - \hat{D}(P\|Q) \geqslant 0$. As this last expression is nondecreasing in $\gamma$, we get that
\begin{align}\label{eqn:bound-Z-gamma-ln3-Z}
    \sup_{0 < \gamma \leqslant \nu} \mbE[Z^{\gamma} \ln^3 Z] \leqslant \frac{1}{\mu^3} e^{\nu ( \hat{D}_{1+\nu}(P \| Q) - \hat{D}(P \| Q))} \ln^3\left(e^{(\mu+\nu)(\hat{D}_{1+\mu+\nu}(P \| Q) - \hat{D}(P\|Q))} + e^2 \right).
\end{align}
This proves that 
\begin{align}
    \hat{D}_{\alpha}(P \| Q) &\leqslant \hat{D}(P \| Q) + \frac{(\alpha - 1)}{2} \hat{V}(P \| Q) +  \frac{(\alpha - 1)^2}{6} (\text{RHS of \eqref{eqn:bound-Z-gamma-ln3-Z}})
\end{align}
and therefore, after converting back to base 2, that
\begin{align}
\label{eq:continuity_alpha_classical}
D_{\alpha}(P \| Q) &\leqslant D(P \| Q) + \frac{(\alpha - 1) \ln 2}{2} V(P \| Q) +  (\alpha - 1)^2 K_{P,Q}(\alpha,\mu)
\end{align}
with $K_{P,Q}(\alpha,\mu) = \frac{1}{6 \mu^3 \ln 2} 2^{(\alpha - 1) ( D_{\alpha}(P \| Q) - D(P \| Q))} \ln^3\left(2^{(\alpha+\mu - 1) (D_{\alpha+\mu}(P \| Q) - D(P\|Q))} + e^2 \right)$.

Now in order to get the general statement, we use the fact that the Petz divergence between states $\rho$ and $\sigma$ is equal to the $\alpha$-divergence of Nussbaum-Szko\l{}a distributions~\cite{ns09}, i.e., for all $\alpha \geqslant 0$
\begin{align*}
D'_{\alpha}(\rho \| \sigma) = D_{\alpha}(P^{[\rho, \sigma]} \| Q^{[\rho, \sigma]}) \ ,
\end{align*}
where $P^{[\rho, \sigma]}(x,y) = \lambda_x | \braket{e_x}{f_y} |^2$ and $Q^{[\rho, \sigma]}(x,y) = \mu_y | \braket{e_x}{f_y} |^2$ where $\{\lambda_x, \ket{e_x}\}_{x}$ are the eigenvalues and eigenvectors of $\rho$ and $\{\mu_y, \ket{f_y}\}_{y}$ are the eigenvalues and eigenvectors of $\sigma$. Note that $P^{[\rho, \sigma]}$ and $Q^{[\rho, \sigma]}$  only depend on $\rho$ and $\sigma$ and not on $\alpha$, and $P^{[\rho, \sigma]}$ and $Q^{[\rho, \sigma]}$ have the same normalization as $\rho$ and $\sigma$, respectively. Note that by taking the limit $\alpha \to 1$, we also get $D(\rho \| \sigma) = D(P^{[\rho, \sigma]} \| Q^{[\rho, \sigma]})$. In addition, by taking the derivative at $\alpha = 1$, we get that $V( \rho \| \sigma) = V(P \| Q)$~\cite[Proposition 4.9]{Tom15book}. Applying inequality~\eqref{eq:continuity_alpha_classical} to $P^{[\rho, \sigma]}$ and $Q^{[\rho, \sigma]}$, we get the desired result.
\end{proof}
To obtain a quantitative continuity for $H_{\alpha}(A|B)_{\rho}$ at $\alpha = 1$, it suffices to use Lemma~\ref{lem_HalphaH_second_order_new} with $\rho = \rho_{AB}$, $\sigma = \id_{A} \otimes \rho_{B}$ together with the fact that $D_{\alpha}(\rho \| \sigma) \leqslant D'_{\alpha}(\rho \| \sigma)$. In addition, to simplify the statement, we set $\mu = 2 - \alpha$.
\begin{corollary}
    \label{cor:continuity-bound-halpha_new}
Let $\rho_{AB}$ be a density operator. Then we have for any $\alpha \in (1,2)$,
\begin{equation*}
    H_{\alpha}(A|B)_{\rho} \geqslant H(A|B)_{\rho} - \frac{(\alpha - 1) \ln 2}{2} V(A|B)_{\rho} - (\alpha - 1)^2 K(\alpha) \ ,
\end{equation*}
where $K(\alpha) = \frac{1}{6 (2-\alpha)^3 \ln 2} \cdot 2^{(\alpha - 1)(- H'_{\alpha}(A | B)_{\rho} + H(A | B)_{\rho} )} \ln^3\left( 2^{- H'_{2}(A | B)_{\rho} + H(A | B)_{\rho}} + e^2 \right)$.
\end{corollary}

\section{Entropy accumulation with improved second order}\label{sec:accumulation}

We start by recalling the framework for the entropy accumulation theorem~\cite{dfr16}. For $i \in \{1,\dots,n\}$, let $\cM_i$ be a TPCP map from $R_{i-1}$ to $X_i A_i B_i R_i$, where $A_i$ is finite-dimensional and where $X_i$ represents a classical value from an alphabet $\cX$ that is  determined by $A_i$ and $B_i$ together. More precisely, we require that, $\cM_{i} = \cT_{i} \circ \cM'_i$ where $\cM'_{i}$ is an arbitrary TPCP map from $R_{i-1}$ to $A_{i} B_{i} R_{i}$  and $\cT_i$ is a TPCP map from $A_{i}B_{i}$ to $X_{i} A_i B_i$ of the form
\begin{align}
\label{eq_form_extraction}
\cT_{i}(W_{A_iB_i}) = \sum_{y \in \cY , z \in \cZ} (\Pi_{A_i,y} \otimes \Pi_{B_i,z}) W_{A_i B_i} (\Pi_{A_i,y} \otimes \Pi_{B_i,z}) \otimes \proj{t(y,z)}_{X_i} \ ,
\end{align}
where $\{\Pi_{A_i, y}\}$ and $\{\Pi_{B_i, z}\}$ are families of mutually orthogonal projectors on $A_i$ and $B_i$, and where $t : \cY \times \cZ  \to \cX$ is a deterministic function. 

The entropy accumulation theorem stated below will hold for states of the form 
\begin{align} \label{eq_rhomap}
  \rho_{A_1^n B_1^n X_1^n E} = \tr_{R_n} ({\cM_n \circ \dots \circ \cM_1} \otimes \cI_E)(\rho^0_{R_0 E})
\end{align}
where $\rho^0_{R_0 E} \in  \mathrm{D}(R_0 \otimes E)$ is a density operator on $R_0$ and an arbitrary system $E$.  In addition, we require that the Markov conditions 
\begin{align} \label{eq_Markovgen}
  A_1^{i-1} \leftrightarrow B_1^{i-1}  E \leftrightarrow B_{i}
\end{align}
be satisfied for all $i \in \{1, \ldots, n\}$; i.e.~$I(A_1^{i-1}; B_i | B_1^{i-1}E)_{\rho} = 0$.

Let $\mbP$ be the set of probability distributions on the alphabet $\cX$ of $X_i$, and let $R$ be a system isomorphic to $R_{i-1}$. For any $q \in \mbP$ we define the set of states 
\begin{align}
\label{eq:def_sigma}
  \Sigma_i(q) = \bigl\{\nu_{X_i A_i B_i R_i R} = (\cM_i \otimes \cI_R)(\omega_{R_{i-1} R}) : \quad  \omega \in \mathrm{D}(R_{i-1} \otimes R) \text{ and } \nu_{X_i} = q \bigr\}  \ ,
\end{align}
where $\nu_{X_i}$ denotes the probability distribution over $\cX$ with the probabilities given by $\bra{x} \nu_{X_i} \ket{x}$. In other words, $\Sigma_i(q)$ is the set of states that can be produced at the output of the channel $\mathcal{M}_i$ and that have a reduced state on the $X_i$ system equal to $q$.
\begin{definition} \label{def_tradeoff}
  A  real function $f$ on $\mbP$ is called a \emph{min-tradeoff function} (or simply tradeoff function for short) for $\cM_i$ if it  satisfies 
\begin{align*}
    f(q) \leqslant \min_{\nu \in \Sigma_i(q)} H(A_i | B_i R)_{\nu} \ .
\end{align*}
Note that if $\Sigma_i(q) = \emptyset$, then $f(q)$ can be chosen arbitrarily.
Our result will depend on some simple properties of the tradeoff function, namely the maximum and minimum of $f$, the minimum of $f$ over valid distributions, and the maximum variance of $f$:
\begin{align*}
\Max{f} &:= \max_{q \in \mathbb{P}} f(q) \\
\Min{f} &:= \min_{q \in \mathbb{P}} f(q) \\
\MinSigma{f} &:= \min_{q : \Sigma_i(q) \neq \emptyset} f(q) \\
\Var{f} &:= \max_{ q : \Sigma_i(q) \neq \emptyset} \sum_{x \in \cX} q(x) f(\delta_{x})^2 - \left(\sum_{x \in \cX} q(x) f(\delta_x) \right)^2 \ ,
\end{align*}
\end{definition}
where $\delta_x$ stands for the distribution with all the weight on element $x$.

We write $\freq{X_1^n}$ for the distribution on $\cX$ defined by $\freq{X_1^n}(x) = \frac{|\{i \in \{1,\dots,n\} : X_i = x\}|}{n}$.  We also recall that in this context, an event $\Omega$ is defined by a subset of $\cX^n$ and we write $\rho[\Omega] = \sum_{x_1^n \in \Omega}\tr(\rho_{A_1^n B_1^n E, x_1^n})$ for the probability of the event $\Omega$ and 
\begin{align*}
  \rho_{X_1^n A_1^n B_1^n E | \Omega} = \frac{1}{\rho[\Omega]} \sum_{x_1^n \in \Omega} \proj{x_1^n} \otimes \rho_{A_1^n B_1^n E, x_1^n}
\end{align*}
for the state conditioned on~$\Omega$.

\begin{theorem}\label{thm:entropyaccumulationext}
    Let $\cM_1,\dots,\cM_n$ and  $\rho_{A_1^n B_1^n X_1^n E}$ be such that~\eqref{eq_rhomap} and the Markov conditions~\eqref{eq_Markovgen} hold, let  $h \in \mathbb{R}$, let $f$ be an affine min-tradeoff function for $\cM_1,\dots,\cM_n$,  and let  $\varepsilon \in (0,1)$. Then, for any event $\Omega \subseteq \cX^n$ that implies $f(\freq{X_1^n}) \geqslant h$,
      \begin{align}
        \label{eqn:eat-min}H_{\min}^{\varepsilon}(A_1^n | B_1^n E)_{\rho_{|\Omega}} & > n h -  c \sqrt{n} - c'
  \end{align}
  holds for 
  \begin{align*}
  c &= \sqrt{2 \ln 2} \left( \log (2d_A^2+1) + \sqrt{2 + \Var{f}} \right) \sqrt{ 1-2\log (\varepsilon \rho[\Omega])} \\
  c' &= \frac{35(1-2\log (\varepsilon \rho[\Omega]))}{\left( \log(2 d_A^2 + 1) + \sqrt{2+\Var{f}} \right)^2} 2^{2 \log d_A + \Max{f} -  \MinSigma{f} } \ln^3\left( 2^{2 \log d_A + \Max{f} -  \MinSigma{f}} + e^2 \right)
  \end{align*}
   where $d_A$ is the maximum dimension of the systems~$A_i$.  
    \end{theorem}
While the above give reasonable bounds in the general case, in order to obtain better finite $n$ bounds in a particular case of interest, we advise the user to instead use the following bound for an $\alpha \in (1,2)$ that is either chosen carefully for the problem at hand or computed numerically: 
    \begin{align}
    \label{eqn:eat-min-min-alpha}
    H_{\min}^{\varepsilon}(A_1^n | B_1^n E)_{\rho_{|\Omega}} &\geqslant
    n h - n  \frac{(\alpha-1) \ln 2}{2} V^2 - \frac{1}{\alpha - 1} \log \frac{2}{\varepsilon^2 \rho[\Omega]^2} - n(\alpha - 1)^{2} K_{\alpha} \ ,
    \end{align}
    with 
    \begin{align}
    V &= \sqrt{\Var{f} + 2} +  \log(2d_A^2+1) \\
\label{eq_Kalpha}
 K_{\alpha} &= \frac{1}{6 (2 - \alpha)^3 \ln 2} \cdot 2^{(\alpha-1)(2 \log d_A + (\Max{f} -  \MinSigma{f}) } \ln^3\left( 2^{2 \log d_A + (\Max{f} -  \MinSigma{f})} + e^2 \right) \ .
\end{align}
Note that in general the optimal choice of $\alpha$ will depend on $n$; in Theorem~\ref{thm:entropyaccumulationext} we have chosen $\alpha$ so that $\alpha-1$ scales as $\Theta(1/\sqrt{n})$, but other choices are possible. As described in the proof, in the case where the systems $A_i$ are classical, we can replace $2 \log d_{A}$ by $\log d_{A}$ in~\eqref{eq_Kalpha}, this comes from the fact that $H_{\alpha}(A_i | C) \geq 0$ whenever $A_i$ is classical but can only be lower bounded by $- \log d_{A}$ in the general case. This bound holds under the exact same conditions as Theorem~\ref{thm:entropyaccumulationext} and for any $\alpha \in (1,2)$, and this is the bound we use to obtain the numerical results presented in the application presented in Section~\ref{sec:DIRE}.
    The choice of $\alpha$ made to get Theorem~\ref{thm:entropyaccumulationext} is not the optimal one, but it was chosen to have a relatively simple expression showing the dependence on the main parameters without optimizing the constants.

    The proof structure is the same as in~\cite{dfr16}.  The only difference is when using the continuity of $D_{\alpha}$, we use the more precise estimate in Lemma~\ref{lem_HalphaH_second_order_new}, and we use the various properties of the entropy variance proven in Section~\ref{sec:divergence-variance} to bound the second-order term.

\begin{proposition}\label{prop:accumulation-alpha}
Let $\cM_1, \ldots, \cM_n$ and $\rho_{A_1^n B_1^n X_1^n E}$ be such that~\eqref{eq_rhomap} and the Markov conditions~\eqref{eq_Markovgen} hold, let  $h \in \mathbb{R}$, and let $f$ be an affine min-tradeoff function $f$ for $\cM,\dots,\cM_n$. Then, for any event $\Omega$ which implies  $f(\freq{X_1^n}) \geqslant h$,
      \begin{align}
          \label{eqn:eat-min-alpha} 
                  H^{\uparrow}_{\alpha}(A_1^n | B_1^n E)_{\rho_{|\Omega}} 
        & > n h - n \frac{(\alpha-1) \ln 2}{2} V^2 - \frac{\alpha}{\alpha - 1} \log \frac{1}{\rho[\Omega]} - n(\alpha - 1)^{2} K_{\alpha}
        \end{align}
        holds for $\alpha$ satisfying $\alpha \in (1,2)$, and $V = \sqrt{\Var{f} + 2} +  \log(2d_A^2+1)$, where $d_A$ is the maximum dimension of the systems~$A_i$ and $K_{\alpha}$ is defined in~\eqref{eq_Kalpha}.
\end{proposition}

\begin{proof}
    
    The first step of the proof is to construct a state that will allow us to lower-bound $H^\uparrow_{\alpha}(A_1^n | B_1^n E)_{\rho_{|\Omega}}$ using a chain rule similar to the one in Corollary~\ref{cor_conditionalmap}, while ensuring that the tradeoff function is taken into account. In order to achieve this, we proceed as in~\cite{dfr16} and introduce an additional $D$ system that can be thought of as an entropy price
     that encodes the tradeoff function. 
     More precisely, for every $i$, let $\mathcal{D}_i : X_i \rightarrow X_i D_i $, be a TPCP map defined as
\begin{align*}
    \mathcal{D}_i(W_{X_i}) = \sum_{x \in \mathcal{X}} \bra{x}W_{X_i} \ket{x} \cdot \proj{x}_{X_i} \otimes \tau(x)_{D_i} \ ,
\end{align*}
    where $\tau(x)$ is such that $H(D_i)_{\tau(x)} = \Max{f} - f(\delta_x)$ (here $\delta_x$ stands for the distribution with all the weight on element $x$). This is possible because $\Max{f} - f(\delta_x) \in [0, \Max{f} - \Min{f}]$ and we choose the dimension of the systems $D_i$ to be equal to $d_D = \left\lceil  2^{\Max{f} - \Min{f}} \right\rceil$. More precisely, we fix $\tau(x)$ to be a mixture between a uniform distribution on $\{1, \dots, \lfloor 2^{\Max{f} - f(\delta_x)} \rfloor \}$ and a uniform distribution on $\{1, \dots, \lceil 2^{\Max{f} - f(\delta_x)} \rceil\}$. We note that compared to~\cite{dfr16}, our choice of state $\tau(x)$ is different. In fact, in~\cite{dfr16}, an additional system $\bar{D}$ was added to the conditioning and $\tau(x)$ was an appropriate mixture of a maximally entangled state on $D \otimes \bar{D}$ and a maximally mixed state on $D \otimes \bar{D}$. This choice is not adapted here because we will need the entropy variance of $\tau(x)$ to be small, for this reason we choose $\tau(x)$ to be basically uniform on a set of size $2^{\Max{f} - f(\delta_x)}$.

    Now, let
    \begin{align*}
      \bar{\rho} := (\mathcal{D}_n \circ \dots \circ \mathcal{D}_1)(\rho) \ . 
    \end{align*}
  Exactly as in the corresponding claim in~\cite{dfr16}, we can relate conditional entropy $H^{\uparrow}_{\alpha}(A_1^n | B_1^n E)_{\rho_{|\Omega}}$ to the conditional entropy of the constructed state $\bar{\rho}$:
   \begin{align}
    H^{\uparrow}_{\alpha}(A_1^n | B_1^n E)_{\rho_{|\Omega}}
        &\geqslant   H^{\uparrow}_{\alpha}(A_1^n D_1^n | B_1^n E)_{\bar{\rho}_{|\Omega}} - n \Max{f} + n h \ . 
    \end{align}
  The next step is to relate the entropies on the conditional state $\rho_{|\Omega}$ to those on the unconditional state. To do this, we use Lemma~\ref{lem:abx-chain-rule-opt} applied to $\bar{\rho} = \rho[\Omega] \bar{\rho}_{|\Omega} + (\bar{\rho} - \rho[\Omega] \bar{\rho}_{|\Omega})$, together with the fact that $H_{\alpha}^{\uparrow} \geqslant H_{\alpha}$, and obtain
   \begin{align}
   \label{eq_lb_using_d}
    H^{\uparrow}_{\alpha}(A_1^n | B_1^n E)_{\rho_{|\Omega}}
          &\geqslant   H_{\alpha}(A_1^n D_1^n | B_1^n E)_{\bar{\rho}} - \frac{\alpha}{\alpha - 1} \log\frac{1}{\rho[\Omega]}  - n \Max{f} + n h \ .
  \end{align}
   
To show the desired inequality~\eqref{eqn:eat-min-alpha}, it now suffices to prove that $H_{\alpha}(A_1^n D_1^n | B_1^n E)_{\bar{\rho}}$ is lower bounded by (roughly) $n \Max{f}$. 

In order to lower bound $H_{\alpha}(A_1^n D_1^n | B_1^n E)_{\bar{\rho}}$, we are now going to use the chain rule for R\'enyi entropies in Corollary~\ref{cor_conditionalmap} $n$ times on the state $\bar{\rho}$, with the following substitutions at step $i$:
    \begin{itemize}
        \item $A_1 \rightarrow A_1^{i-1} D_1^{i-1}$
        \item $B_1 \rightarrow B_1^{i-1} E $
        \item $A_2 \rightarrow A_i D_i$
        \item $B_2 \rightarrow B_i$.
    \end{itemize} 
    To check that the Markov chain condition holds, observe that $I(A_1^{i-1} D_1^{i-1} : B_i | B_1^{i-1} E) = I(A_1^{i-1} : B_i | B_1^{i-1} E) + I(D_1^{i-1} : B_i | B_1^{i-1} E A_1^{i-1})$. Using~\eqref{eq_Markovgen}, we have that $I(A_1^{i-1} : B_i | B_1^{i-1} E) = 0$ and as $D_1^{i-1}$ is determined by $A_1^{i-1} B_1^{i-1}$, we also have $I(D_1^{i-1} : B_i | B_1^{i-1} E A_1^{i-1}) = 0$.
    Thus, Corollary~\ref{cor_conditionalmap} gives
    \begin{align}
    \nonumber
    &H_{\alpha}(A_1^n D_1^n | B_1^n E )_{\bar{\rho}} \\
         &\geqslant \sum_i \inf_{\omega_{R_{i-1} R}} H_{\alpha}(A_i D_i|B_i R)_{(\mathcal{D}_i \circ \cM_i)(\omega)} \nonumber \\
         & \geqslant \sum_i \inf_{\omega_{R_{i-1} R}} \left(  H(A_i D_i|B_i R)_{(\mathcal{D}_i \circ \cM_i)(\omega)} - \frac{(\alpha-1) \ln 2}{2} V(A_iD_i | B_i  R_i)_{(\mathcal{D}_i \circ \cM_i)(\omega)} -  (\alpha-1)^2 K(\alpha) \right) \label{eq:exp_after_continuity} \ ,
    \end{align}
    where we have invoked Corollary \ref{cor:continuity-bound-halpha_new} in the second inequality. Here,  
    \begin{align*}
        K(\alpha) = \frac{1}{6 (2 - \alpha)^3 \ln 2} \cdot 2^{(\alpha - 1)(- \eta_1 +  \eta_0)} \ln^3\left( 2^{- \eta_2 + \eta_0} + e^2 \right) \ .
    \end{align*}
    with $\eta_1 = H'_{\alpha}(A_i D_i | B_i R_i)_{(\mathcal{D}_i \circ \cM_i)(\omega)}$, $\eta_0 = H(A_i D_i | B_i R_i)_{(\mathcal{D}_i \circ \cM_i)(\omega)}$ and $\eta_2 = H'_{2}(A_i D_i | B_i R_i)_{(\mathcal{D}_i \circ \cM_i)(\omega)}$.

    For any such state $\omega_{R_{i-1} R}$, we have
    \begin{align*}
        H(A_i D_i | B_i R)_{(\mathcal{D}_i \circ \cM_i)(\omega)}
        & = H(A_i X_i D_i | B_i R)_{(\mathcal{D}_i \circ \cM_i)(\omega)} \\
        &= H(A_i X_i | B_i R)_{\cM_i(\omega)}  + H(D_i|X_i)_{(\mathcal{D}_i \circ \cM_i)(\omega)}\\
        &= H(A_i  | B_i R)_{\cM_i(\omega)}  + \sum_x q(x) H(D_i)_{\tau(x)}\\
         & = H(A_i | B_i R)_{\cM_i(\omega)}  + \sum_x q(x) \bigl( \Max{f} - f(\delta_x) \bigr) \\
         &= H(A_i | B_i R)_{\cM_i(\omega)} + \Max{f} - f(q)\ ,
        \end{align*}
        where $q = \cM_i(\omega)_{X_i}$ denotes the distribution of $X_i$ on $\mathcal{X}$ obtained from the state $\cM_i(\omega)$. The third equality comes from the fact that $X_i$ is determined by $A_iB_i$. The last equality holds because $f$ is affine. Using the fact that $f$ is a min-tradeoff function, we get that $H(A_i|B_i R)_{\mathcal{M}_i(\omega)} \geqslant f(q)$ and therefore:
    \begin{align*}
	\Max{f} \leqslant        H(A_i D_i | B_i R)_{(\mathcal{D}_i \circ \cM_i)(\omega)} \leqslant \log d_{A_i} + \Max{f} - f(q) \ .
        \end{align*}
        The lower bound allows us to lower bound the first term in Eq.~\eqref{eq:exp_after_continuity}. The upper bound will allow us to bound the last term in Eq.~\eqref{eq:exp_after_continuity}. In fact, as the systems $D_i$ are classical, we have $\eta_1, \eta_2 \geqslant - \log d_{A_i}$ by Lemma~\ref{lem:hprimealpha-part-classical-dim-bound} (and in the case where $A_i$ are classical, we have $\eta_1, \eta_2 \geqslant 0$) and thus
\begin{align*}
    K(\alpha)
&\leqslant \frac{1}{6 (2-\alpha)^3 \ln 2} \cdot 2^{(\alpha - 1)(\log d_{A} +  \log d_{A} + \Max{f} - f(q))} \ln^3\left( 2^{\log d_{A} + \log d_{A} + \Max{f} - f(q)} + e^2 \right) \\
&\leqslant \frac{1}{6 (2-\alpha)^3 \ln 2} \cdot 2^{(\alpha - 1)(2\log d_{A}  + \Max{f} - \MinSigma{f})} \ln^3\left( 2^{2\log d_{A} + \Max{f} - \MinSigma{f}} + e^2 \right) \ ,
\end{align*}
as by definition $\Sigma_i(q)$ is not empty (it contains $\cM_i(\omega)$).

We now analyze the second term of Eq.~\eqref{eq:exp_after_continuity}. Using Lemma~\ref{lem_det_func_var} and then Lemma~\ref{lem_var_decomp} we have
\begin{align}
V(A_iD_i | B_i R_i)_{(\mathcal{D}_i \circ \cM_i)(\omega)}
&= V(A_i X_i D_i | B_i R_i)_{(\mathcal{D}_i \circ \cM_i)(\omega)} \\
\label{eqn:V-inequality}&\leqslant \left(\sqrt{V(A_i | B_i R_i)_{(\mathcal{D}_i \circ \cM_i)(\omega)}} + \sqrt{V(D_i | X_i)_{(\mathcal{D}_i \circ \cM_i)(\omega)}}\right)^2 \ .
\end{align}

We bound the first term by the dimension of $A$ using Corollary~\ref{lem_var_dim}.
\begin{align*}
V(A_i | B_i R_i)_{(\mathcal{D}_i \circ \cM_i)(\omega)} \leqslant \log^2 (2d^2_A+1).
\end{align*} 
 For the second term, using the notation $q = \cM_i(\omega)_{X_i}$, we have using Lemma~\ref{lem:divergence-variance-classical-X}
\begin{align*}
V(D_i | X_i)_{(\mathcal{D}_i \circ \cM_i)(\omega)} &= \sum_{x \in \cX} q(x) V(D_{i})_{\tau(x)} + \Variance(W) \ ,
\end{align*}
where $W$ takes the value $H(D_i)_{\tau(x)} = \Max{f} - f(\delta_{x})$ with probability $q(x)$. We have
\begin{align*}
\Variance(W) &= \sum_{x \in \cX} q(x) (\Max{f} - f(\delta_{x}))^2 - \left(\sum_{x} q(x) (\Max{f} - f(\delta_x)) \right)^2 \\
&\leqslant \sup_{\omega, q = \cM_i(\omega)_{X_i}} \sum_{x \in \cX} q(x) f(\delta_{x})^2 - \left(\sum_{x} q(x) f(\delta_x) \right)^2 \\
&\leqslant \Var{f} \ .
\end{align*}
To bound $V(D_i)_{\tau(x)}$ recall that $\tau(x)$ is a mixture between the uniform distribution on $\{1, \dots, \lfloor 2^{\Max{f} - f(\delta_{x})} \rfloor\}$ and the uniform distribution on $\{1, \dots, \lceil 2^{\Max{f} - f(\delta_{x})} \rceil\}$. Note that if $2^{\Max{f} - f(\delta_{x})}$ is an integer, $\tau(x)$ is uniformly distributed and thus $V(D_i)_{\tau(x)} = 0$. Assuming $2^{\Max{f} - f(\delta_{x})}$ is not an integer, let $\lfloor 2^{\Max{f} - f(\delta_{x})} \rfloor = k$. Then $\tau(x)$ is a distribution on $\{1, \dots, k+1\}$ and we have for some $p$ and $p' \leq p$, $\bra{j} \tau(x) \ket{j} = p$ for all $j \in \{1, \dots, k\}$ and $\bra{k+1} \tau(x) \ket{k+1} = p'$. The normalization condition is $kp + p' = 1$ and thus, $p' = 1 - kp$.
    We can now observe that the entropy variance $V(D_i)_{\tau(x)}$ is simply $\Variance(-\log p  + Z) = \Variance(Z)$, where $Z$ is a random variable that is equal to 0 with probability $1-p'$ and to $\log p - \log p'$ with probability $p'$. This variance can then be computed as
    \begin{align*}
        \Variance(Z) &= \mbE[Z^2] - \mbE[Z]^2\\
        &= p'(\log p - \log p')^2 - {p'}^2(\log p - \log p')^2\\
        &= p'(1-p') \log^2 \left( \frac{p}{p'} \right).
    \end{align*}
    Now, we use the fact that $\log^2 z \leqslant 2z$ and continue:
    \begin{align*}
        \Variance(Z) &\leqslant 2p'(1-p') \frac{p}{p'} \leqslant 2.
    \end{align*}


As a result,
\begin{align*}
V(D_i | X_i)_{(\mathcal{D}_i \circ \cM_i)(\omega)} 
		&\leq 2 + \Var{f} \ .
\end{align*}

Putting everything together, Eq.~\eqref{eq_lb_using_d} becomes
   \begin{align*}
        H^{\uparrow}_{\alpha}(A_1^n | B_1^n E)_{\rho_{|\Omega}} &\geqslant n h - n \frac{(\alpha-1) \ln 2}{2} \left(\log (2d^2_A+1) + \sqrt{2 + \Var{f}} \right)^2 \\
         &\quad\quad\quad\quad - n(\alpha-1)^2 K_{\alpha} -  \frac{\alpha}{\alpha-1} \log\frac{1}{\rho[\Omega]} \ .
   \end{align*}  
\end{proof}

Theorem \ref{thm:entropyaccumulationext} is then obtained from Proposition~\ref{prop:accumulation-alpha} by choosing $\alpha$ appropriately.
\begin{proof}[Proof of Theorem \ref{thm:entropyaccumulationext}]
    We start by lower-bounding the smooth min-entropy by a Rényi entropy: for $\alpha \in (1,2]$ (see e.g.,~\cite[Proposition 6.5]{Tom15book}), we have
    \begin{align}\label{eq_eatblock}
        H_{\min}^{\varepsilon}(A_1^n | B_1^n E)_{\rho_{|\Omega}} \geqslant H^{\uparrow}_{\alpha}(A_1^n | B_1^n E)_{\rho_{|\Omega}} - \frac{\log(2/\varepsilon^2)}{\alpha-1} \ .
    \end{align}
    Then Proposition \ref{prop:accumulation-alpha} yields for $\alpha \in (1, 1 + \frac{1}{2 \ln 2})$
    \begin{align}\label{eq_eathmin_halpha}
    H_{\min}^{\varepsilon}(A_1^n | B_1^n E)_{\rho_{|\Omega}} &> n h - n  \frac{(\alpha-1) \ln 2}{2} V^2 - \frac{\alpha}{\alpha - 1} \log \frac{1}{\rho[\Omega]} - n(\alpha - 1)^{2} K_{\alpha} - \frac{\log (2/\varepsilon^2)}{\alpha-1} \\
    &\geqslant n h - n  \frac{(\alpha-1) \ln 2}{2} V^2 - \frac{1}{\alpha - 1} \log \frac{2}{\varepsilon^2 \rho[\Omega]^2} - n(\alpha - 1)^{2} K \ , \nonumber
    \end{align}
    where for the first inequality, $K_{\alpha}$ is as in~\eqref{eq_Kalpha} and in the second inequality, we used the fact that $\alpha \leqslant 1 + \frac{1}{2 \ln 2}$ and defined
    \begin{align*}
    K = 12 \cdot 2^{(2 \log d_A + (\Max{f} -  \MinSigma{f}) } \ln^3\left( 2^{2 \log d_A + (\Max{f} -  \MinSigma{f})} + e^2 \right)\ .
    \end{align*}
     To make the terms in $\alpha - 1$ and $\frac{1}{\alpha-1}$ match, we 
 choose
    \begin{align} \label{eq_alphachoiceext}
       \alpha := 1 + \frac{ \sqrt{ 2 \log \frac{2}{\rho[\Omega]^2 \varepsilon^2} }}{\sqrt{n \ln 2} V} \ .
    \end{align}
    Assuming that $n \geqslant\frac{8 \ln 2 \log \frac{2}{\varepsilon^2 \rho[\Omega]^2}}{V^2}$ to have $\alpha \leqslant 1 + \frac{1}{2\ln 2}$, we obtain
    \begin{equation}\label{eqn:hmin-bound-alpha-subbed-in}
         H_{\min}^{\varepsilon}(A_1^n | B_1^n E)_{\rho} > n h - \sqrt{n} V \sqrt{ (2 \ln 2) \log \frac{2}{\rho[\Omega]^2 \varepsilon^2}} -  \frac{2\log \frac{2}{\rho[\Omega]^2 \varepsilon^2} }{V^2 \ln 2} K .
     \end{equation}
         Note that if $n < \frac{8\ln 2 \log \frac{2}{\varepsilon^2 \rho[\Omega]^2}}{V^2}$ then 
         \[
             \sqrt{n} V \sqrt{ (2 \ln 2) \log \frac{2}{\rho[\Omega]^2 \varepsilon^2}} > \frac{1}{2} n V^2.
         \]
         As we may assume $h \leq \log d_A$ (otherwise the event $\Omega$ will have zero probability) and using the definition of $V$, we have that 
         \begin{align*}
             n \left(h- \frac{1}{2}V^2\right) &\leqslant n \log d_{A} - \frac{1}{2} (\sqrt{2} + \log(2d_{A}^2+1))^2  n\\
             &\leqslant - n \log d_{A}\\
         \end{align*}
which implies that~\eqref{eqn:eat-min} is true in a trivial way.
\end{proof}

\subsection{EAT channels with infrequent sampling}\label{sec:infrequent-sampling}

This section can be seen as a user guide to apply the entropy accumulation result presented here in the very common setting where the ``testing'' is only done in a few rounds that are sampled at random. From the entropy accumulation point of view, the reason for testing is to restrict the optimization involved in the tradeoff function to states $\omega_{R_{i-1} R}$ satisfying the output statistics~\eqref{eq:def_sigma}, e.g., winning the CHSH game with a certain probability. However, testing can be costly in terms of randomness or rate and for this reason, the probability of testing, denoted $\gamma$ is often chosen to be small.
We start by defining ``channels with infrequent sampling'', which formalizes the concept of a protocol in which we test only a few positions:

\begin{definition}[Channel with infrequent sampling]
    A channel with testing probability $\gamma \in [0,1]$ is an EAT channel $\mathcal{M}_{i,R_{i-1} \rightarrow X_i A_i B_i R_i}$ such that $\mathcal{X} = \mathcal{X}' \cup \{ \bot \}$ and that can be expressed as
    \[ \mathcal{M}_{i, R_{i-1} \rightarrow X_i A_i B_i R_i}(\cdot) = \gamma \mathcal{M}^{\mathrm{test}}_{i,R_{i-1} \rightarrow X_i A_i B_i R_i}(\cdot) + (1-\gamma) \mathcal{M}^{\mathrm{data}}_{i,R_{i-1} \rightarrow A_i B_i R_i}(\cdot) \otimes \proj{\bot}_{X_i}, \]
    where $\mathcal{M}_i^{\mathrm{test}}$ never outputs the symbol $\bot$ on $X_i$.
%
\end{definition}

The following lemma gives a general way of constructing a tradeoff function $f$ for the map $\mathcal{M}_{i}$ using a sort of ``crossover'' tradeoff function $g$ for the map $\mathcal{M}_i$ but using the statistics from $\mathcal{M}_i^{\mathrm{test}}$ only. More precisely, the function $g$ is defined by restricting the input of the map $\mathcal{M}_{i}$ to be ones that are consistent with the output statistics given by the map $\mathcal{M}^{\mathrm{test}}_{i}$. The lemma also gives general bounds on the relevant properties of $f$ as a function of $\gamma$ and simple properties of $g$.
\begin{lemma}
\label{lem:infrequent_sampling}
    Let $\mathcal{M}_i = \cM_{R_{i-1} \rightarrow X_i A_i B_i R_i}$ be a channel with testing probability $\gamma$ as defined above. 
Assume that the affine function $g : \mbP(\cX') \to \mbR$ satisfies for any $q \in \mbP(\cX')$
\begin{align}
\label{eq:tradeoff_g}
g(q') \leqslant \min_{\omega \in \mathrm{D}(R_{i-1} \otimes R)} \{ H(A_i | B_i R)_{(\cM_i \otimes \cI_R)(\omega_{R_{i-1} R})} :  \left((\cM^{\mathrm{test}}_i \otimes \cI_R)(\omega_{R_{i-1} R})\right)_{X_i} = q' \bigr\}. 
\end{align}
Note that if the set $\{ \omega \in \mathrm{D}(R_{i-1} \otimes R) : \left((\cM^{\mathrm{test}}_i \otimes \cI_R)(\omega_{R_{i-1} R})\right)_{X_i} = q' \}$ is empty, the minimum is set to $+\infty$ or in other words, there is no constraint on $g(q')$.
    Then, the affine function $f : \mbP(\cX) \to \mbR$ defined by
     \begin{align*}
        f(\delta_x) &=  \Max{g} + \frac{1}{\gamma} (g(\delta_x) - \Max{g}) \;\;\;\;\; \forall x \in \mathcal{X}'\\
        f(\delta_{\bot}) &= \Max{g} 
    \end{align*}   
    is a min-tradeoff function for $\mathcal{M}_{i}$. Moreover,
    \begin{align*}
    \Max{f} &= \Max{g} \\
    \Min{f} &= \left(1-\frac{1}{\gamma} \right) \Max{g} + \frac{1}{\gamma}\Min{g} \\
    \MinSigma{f} &\geqslant\Min{g} \\
     \Var{f} &\leq \frac{1}{\gamma}\left(\Max{g} - \Min{g} \right)^2.
     \end{align*}
    \label{lem:f-sampling-prop}
\end{lemma}
\begin{proof}
The value for $\Min{f}$ and $\Max{f}$ follow directly from the definition.

To prove that $f$ is a tradeoff function for $\cM_i$, we first determine $\Sigma_i(q)$ (see Definition~\ref{def_tradeoff}). If $q$ is not of the form $q(x) = \gamma q'(x)$ when $x \in \cX'$ and $q(\bot) = (1-\gamma)$ for some $q' \in \mbP(\cX')$, then we know that $\Sigma_i(q) = \emptyset$. So it suffices to focus on distributions $q$ that have this form. Then we have
\begin{align*}
    f(q) &= \sum_{x \in \mathcal{X}'} q(x)\left( \Max{g} + \frac{1}{\gamma}\left( g(\delta_x) - \Max{g} \right) \right) + (1-\gamma) \Max{g}\\
&= \Max{g} + \sum_{x \in \cX'} q'(x) (g(\delta_x) - \Max{g}) \\
&= g(q') \ .
\end{align*}
Using the condition~\eqref{eq:tradeoff_g}, we get
\begin{align*}
f(q) = g(q') &\leq \min_{\omega \in \mathrm{D}(R_{i-1} \otimes R)} \{ H(A_i | B_i R)_{(\cM_i \otimes \cI_R)(\omega_{R_{i-1} R})} :  \left((\cM^{\mathrm{test}}_i \otimes \cI_R)(\omega_{R_{i-1} R})\right)_{X_i} = q' \bigr\} \\ 
&\leq \min_{\nu \in \Sigma_i(q)} H(A_i | B_i R)_{\nu} \ ,
\end{align*}
where for the last inequality, we used the fact that for a $\nu \in \Sigma_i(q)$, there exists an $\omega_{R_{i-1} R}$ such that $(\cM_i \otimes \cI_R)(\omega_{R_{i-1} R}) = \nu$ and $\left((\cM^{\mathrm{test}}_i \otimes \cI_R)(\omega_{R_{i-1} R})\right)_{X_i} = q'$. Thus, $f$ is a min-tradeoff function.

Now for $\MinSigma{f}$, we have
\begin{align*}
\MinSigma{f} 
&= \min_{q : \Sigma_i(q) \neq \emptyset} f(q) \\
&\geqslant\min_{q' \in \mbP(\cX')} g(q') \\
&= \Min{g} \ .
\end{align*}
Finally, for the variance, we have for $q$ such that $\Sigma_i(q) \neq \emptyset$,
\begin{align*}
&\sum_{x \in \cX} q(x) \left(f(\delta_{x}) - \sum_{x \in \cX} q(x) f(\delta_x) \right)^2 \\
&= \sum_{x \in \cX'} \gamma q'(x) \left(\Max{g} + \frac{1}{\gamma}(g(\delta_x) - \Max{g}) - g(q') \right)^2 + (1-\gamma) (\Max{g} - g(q'))^2 \\
&= \frac{1}{\gamma} \sum_{x \in \cX'} q'(x) \left(  (\Max{g} - g(\delta_x)) - \gamma (\Max{g} - g(q')) \right)^2 + (1-\gamma) (\Max{g} - g(q'))^2 \ .
\end{align*}
We can expand the first term and get
\begin{align*}
&\frac{1}{\gamma} \sum_{x \in \cX'} q'(x) \left(  (\Max{g} - g(\delta_x)) - \gamma (\Max{g} - g(q')) \right)^2 \\
&=  \sum_{x \in \cX'} \frac{q'(x)}{\gamma} \left( (\Max{g} - g(\delta_x))^2 - 2 \gamma (\Max{g} - g(\delta_x)) (\Max{g} - g(q')) + \gamma^2 (\Max{g} - g(q'))^2 \right)  \\
&=  \sum_{x \in \cX'} \frac{q'(x)}{\gamma} (\Max{g} - g(\delta_x))^2 - 2(\Max{g} - g(q'))^2 + \gamma (\Max{g} - g(q'))^2 \\
&\leqslant \frac{1}{\gamma} \left( \Max{g} - \Min{g} \right)^2 - (2 - \gamma) (\Max{g} - g(q'))^2.
\end{align*}
As a result,
\begin{align*}
\Var{f} \leqslant \frac{1}{\gamma} \left( \Max{g} - \Min{g} \right)^2.
\end{align*}
\end{proof}

Applying Theorem~\ref{thm:entropyaccumulationext} for a map with infrequent sampling, we get a lower bound on the min-entropy of the following form:
\begin{align*}
H_{\min}^{\eps} \geqslant n h - c_1 \sqrt{\frac{n}{\gamma}} - c_2 \ ,
\end{align*}
where $c_1$ and $c_2$ are constants that only depend on $\varepsilon, \rho[\Omega], d_{A}$ and the properties of $g$ but not on $n$ or the testing probability $\gamma$ (in the expression of $c'$ in Theorem~\ref{thm:entropyaccumulationext} the variance $\Var{f}$ can always be lower bounded by $0$). Note that such a bound will be non-trivial as soon as $\gamma \geqslant \frac{c}{n}$ for some constant $c$ (which corresponds to testing a constant number of rounds). This is to be contrasted with the original entropy accumulation theorem~\cite{dfr16} that instead gives a bound of the form $nh - c_1 \frac{\sqrt{n}}{\gamma} - c_2$ and hence will give a trivial bound when $\gamma = o\left(\frac{1}{\sqrt{n}} \right)$.

\section{Sample application: Device-independent randomness expansion}\label{sec:DIRE}
We now apply our result to one of the main problems to which the original EAT was applied, namely randomness expansion~\cite{Col06,CK11,Pir10,VV11,MS14b,arv16}. This was done using the original EAT in~\cite{arv16}, and, to simplify matters, the protocol we will consider here will be essentially the same. The basic task is the following: we are given a pair of devices from a malicious manufacturer; these devices might have been preprogrammed arbitrarily by the manufacturer, but once we have them, they cannot communicate back to the manufacturer. Our goal is to use those devices to generate a uniformly random string, independent from any other data in the universe, and in particular independent from the quantum data the manufacturer might have kept about our devices. It turns out to be impossible to do this without having a little bit of randomness to begin with, but it is possible to expand a small random string into a much longer one.

We give a security proof for the DI-RE protocol based on the CHSH game described in the box below. Recall that the CHSH game works as follows: a referee chooses uniformly random bits $X$ and $Y$ as inputs for the two devices, and the two devices must respond with $A, B \in \{0,1\}$ respectively without communicating with each other after the questions have been received. The devices win the game if $A \XOR B = XY$ and lose otherwise. The best winning probability for devices using a classical strategy is $3/4$, while the optimal quantum strategy wins with probability $\cos^2(\pi/8) \approx 0.85$. In~\cite[Equation (12)]{arv16} (based on \cite[Section 2.3]{pabgs09}), they give a bound on the amount of randomness produced by the devices assuming that they are using a strategy that allows them to win with probability at least $\omega$; this bound is given by:
\begin{equation}
    H(AB|TE,X=x,Y=y) \geqslant g^*(\omega) := 1 - h\left( \frac{1}{2} + \frac{1}{2} \sqrt{16 \omega (\omega - 1) + 3} \right)
    \label{eqn:pabgs09}
\end{equation}
for any inputs $x,y \in \{0,1\}$ and $\omega \in [\frac{3}{4},\cos^2(\pi/8)]$. This bound is zero at $\omega = 3/4$, one at $\omega = \cos^2(\pi/8)$, and becomes nontrivial as soon as $\omega > 3/4$. The devices are initialized in an arbitrary state by the manufacturer, and at every round of the protocol, we play the game with the devices. To ensure that only a small amount of randomness is consumed by the process of generating the inputs, we randomly choose a small number of test positions (by generating a bit $T$ equal to 1 for test rounds and 0 otherwise), and generate $X$ and $Y$ uniformly at random only for those positions. For the other rounds (that we call the ``data'' rounds), we always fix the inputs to $X = 0$ and $Y= 0$. In the parameter estimation step of the protocol, the number of test rounds for which $A \XOR B = XY$ is computed. For mathematical convenience, we will choose the positions of the test rounds in an iid manner; i.e.~each individual round will have a probability $\gamma$ of being a test round.

\begin{figure}[h]
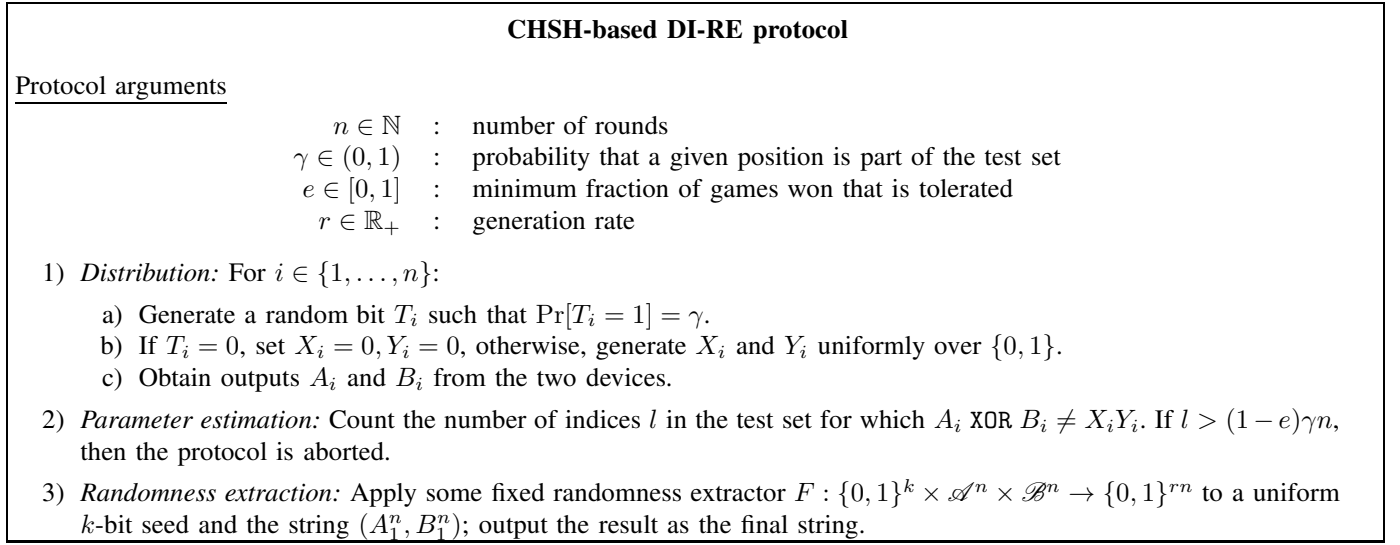

\begin{center}
\fbox{\begin{minipage}{0.96\linewidth}
\smallskip
\begin{center}
\textbf{CHSH-based DI-RE protocol}
\end{center}

\smallskip

\underline{Protocol arguments}
  \vspace{-0.8ex}
\begin{center}
\begin{tabular}{r c l}
  $n \in \mathbb{N}$ & : & number of rounds \\
  $\gamma \in (0,1)$ & : & probability that a given position is part of the test set \\
  $e \in [0,1]$ & : & minimum fraction of games won that is tolerated\\
  $r \in \mathbb{R}_+$ & : & generation rate 
\end{tabular}
\end{center}

\vspace{0.1ex}

\begin{enumerate}
    \setlength{\itemsep}{0.4ex}
    \setlength{\parskip}{0.4ex}
    \setlength{\parsep}{0.4ex} 
    \item \emph{Distribution:} For $i \in \{1, \ldots, n\}$:
        \begin{enumerate}
            \item Generate a random bit $T_i$ such that $\Pr[T_i = 1] = \gamma$. 
            \item If $T_i = 0$, set $X_i = 0, Y_i = 0$, otherwise, generate $X_i$ and $Y_i$ uniformly over $\{0,1\}$.
            \item Obtain outputs $A_i$ and $B_i$ from the two devices.
        \end{enumerate}
  \item \emph{Parameter estimation:}  Count the number of indices $l$ in the test set for which $A_i \XOR B_i \neq X_i Y_i$. If $l > (1-e) \gamma n$, then the protocol is aborted.

  \item \emph{Randomness extraction:} Apply some fixed randomness extractor $F: \{0,1\}^k \times \mathscr{A}^{n} \times \mathscr{B}^n \rightarrow \{0,1\}^{rn}$ to a uniform $k$-bit seed and the string $(A_1^n, B_1^n)$; output the result as the final string.
  \end{enumerate} 
\vspace{-1ex}
\end{minipage} \hspace{1em}}
\end{center} 
\caption{Description of the CHSH-based DI-RE protocol.}
\label{fig:di-re-protocol}
\end{figure}

We model the behavior of the devices as follows. We let $\sigma_{ME}$ be the initial state of the device, $M$ is the system that represents the internal memory of the devices, and $E$ is some reference system that may be in the possession of the manufacturer. Now, let $\cM_i : M \rightarrow M T_i X_i Y_i A_i B_i$ be the TPCP map that is applied by the devices in round $i$. We assume that each of these is of the form depicted in Figure~\ref{fig:di-mi-circuit}, with the position subscript $i$ added to the appropriate systems. The state at the end of step 2 of the protocol is thus:
    \begin{align*} \rho_{M T_1^n X_1^n Y_1^n A_1^n B_1^n E} = \left( \cM_n \circ \dots \circ \cM_1 \right)(\sigma_{ME}) \ ,
  \end{align*}
    and we have computed
    \begin{align*}
    l := |\{i : T_i = 1, A_i \XOR B_i \neq X_i Y_i \}| \ .
    \end{align*}
    Furthermore, we define $\Omega$ as the event that we do not abort after step 2; or, in other words, it is the event that $l \leqslant (1-e) \gamma n$. To apply the entropy accumulation theorem to this setting, we need a min-tradeoff function for the $\mathcal{M}_i$'s. Since Theorem \ref{thm:entropyaccumulationext} demands an affine tradeoff function, the natural choice is to pick the tangent to $g^*$ in \eqref{eqn:pabgs09} at a suitably chosen point $\omega \in (\frac{3}{4}, \cos^2(\frac{\pi}{8}))$. Note that we must also check that the tradeoff function is defined appropriately for all possible distributions we might observe.\footnote{For instance, we might observe a winning rate strictly above $\cos^2(\pi/8)$ on the testing rounds: if the true winning probability of the devices is very close to optimal, then statistical fluctuations might push us slightly over the edge.} 

\begin{figure}[h]
    \centering
    \begin{tikzpicture}[thick]
        \draw  
            (0,0) node[porte, minimum height=1.2cm] (c) {$\mathcal{C}$}
            (c) ++(4, 0) node[porte, minimum height=1.2cm] (d) {$\mathcal{D}$}
            (c) ++(2, 2) node[porte, minimum height=1.2cm] (a) {$\mathcal{A}$}
            (c) ++(2, -2) node[porte, minimum height=1.2cm] (b) {$\mathcal{B}$}
            ;
        \draw   
            (c.east) ++(0, .3) coordinate (c-out1)
            (c.east) ++(0, -.3) coordinate (c-out2)
            (a.west) ++(0, .3) coordinate (a-in1)
            (a.west) ++(0, -.3) coordinate (a-in2)
            (a.east) ++(0, .3) coordinate (a-out1)
            (a.east) ++(0, -.3) coordinate (a-out2)
            (b.west) ++(0, .3) coordinate (b-in1)
            (b.west) ++(0, -.3) coordinate (b-in2)
            (b.east) ++(0, .3) coordinate (b-out1)
            (b.east) ++(0, -.3) coordinate (b-out2)
            (d.west) ++(0, .3) coordinate (d-in1)
            (d.west) ++(0, -.3) coordinate (d-in2)
            ;
        \draw
            (d.east) ++(1.5, 0) coordinate (rightwall)
            (c.west) ++(-1.5, 0) coordinate (leftwall)
            ;
        \draw
            (c.west -| leftwall) node[left] {$M$} to (c.west)
            (c-out1) to[out=0, in=180] node[above left] {$S_A$} (a-in2)
            (c-out2) to[out=0, in=180] node[above right] {$S_B$} (b-in1)
            (a-out2) to[out=0, in=180] node[above right] {$S_A$} (d-in1)
            (b-out1) to[out=0, in=180] node[above left] {$S_B$} (d-in2)
            (a-in1) ++(-1,0) node[left] {$X$} to (a-in1)
            (b-in2) ++(-1,0) node[left] {$Y$} to (b-in2)
            (a-out1) to (a-out1 -| rightwall) node[right] {$A$}
            (b-out2) to (b-out2 -| rightwall) node[right] {$B$}
            (d.east) to (d.east -| rightwall) node[right] {$M$}
            (a-in1) ++(-.5, 0) to ++(.2, .5) coordinate (atmp) to (atmp -| rightwall) node[right] {$X$}
            (b-in2) ++(-.5, 0) to ++(.2, -.5) coordinate (btmp) to (btmp -| rightwall) node[right] {$Y$}
            (a-in1) ++(-1, 1) node[left] (t) {$T$}  to (t -| rightwall) node[right] {$T$}
            ;
        \node[draw, dashed, fit=(a) (b) (c) (d) (t) (btmp)] {};
    \end{tikzpicture}
    \caption{Circuit diagram of $\cM: M \rightarrow MTXYAB$. For every round of the protocol, a circuit of this form is applied, where $\mathcal{A}$ and $\mathcal{B}$ are arbitrary TPCP maps with classical output systems $A$ and $B$, respectively, $\mathcal{C}$ and $\mathcal{D}$ are arbitrary TPCP maps, $T$ is a bit equal to 1 with probability $\gamma$, and $X$ and $Y$ are generated uniformly at random whenever $T=0$, and are fixed to $0,0$ otherwise.}
    \label{fig:di-mi-circuit}
\end{figure}
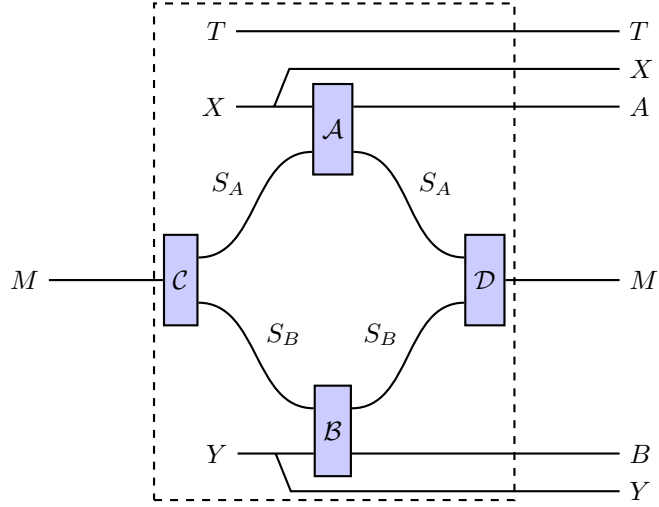

We are now going to use entropy accumulation to prove Theorem~\ref{thm:di-re-security} below, which gives a bound on the randomness generation rate~$r$, i.e., the ratio of uniform bits that can be generated per round of CHSH. 
To get a feeling for the sort of entropy production rates that can be expected of this protocol, we have plotted the final rate obtained (i.e.~the lower bound on $\frac{1}{n} H_{\min}^{\varepsilon}(A_1^n B_1^n | E T_1^n X_1^n Y_1^n)$) as a function of the number of rounds $n$ when we fix the threshold $e$ to 0.8, and when we vary the sampling probability $\gamma$. The result is in Figure~\ref{fig:rate-plots}. We note that the bounds in the figure are not obtained using the bound stated in Theorem~\ref{thm:entropyaccumulationext} directly but rather we used \eqref{eqn:eat-min-min-alpha} with an $\alpha$ optimized numerically for each point on the curve.


\begin{theorem}\label{thm:di-re-security}
    For any device fulfilling the above conditions and for any $\varepsilon \in (0,1)$, testing probability $\gamma \in (0,1)$ and $\frac{3}{4} < e < \cos^2(\frac{\pi}{8})$, after step 2 of the CHSH-based DI-RE protocol, it is the case that either:
    \begin{enumerate}
        \item The min-entropy of $\rho_{|\Omega}$ satisfies:
            \[ H_{\min}^{\varepsilon}(A_1^n B_1^n|E T_1^n X_1^n Y_1^n)_{\rho_{|\Omega}} > ng^*(e) - \sqrt{\frac{n}{\gamma}} c - c', \]
            where 
            $c = \sqrt{2 \ln 2} \left( \log 33 + \sqrt{2 + \frac{1}{\gamma} \left(\frac{dg^*}{d\omega}(e)\right)^2} \right)\sqrt{1 - 4 \log (\varepsilon)}$
            and $c'$ is a constant only depending on $\varepsilon$ and $\frac{dg^*}{d\omega}(e)$, 
            or
        \item The protocol aborts with probability at least $1-\varepsilon$.
    \end{enumerate}
\end{theorem}
First note that applying a Chernoff bound, it is simple to see that provided $e < \cos^2(\pi/8)$, there exist devices that abort the protocol with probability $2^{-\Omega(\gamma n)}$. In addition, provided one is in the first case, one can obtain a secure random string of length roughly $ng^*(e)$ by choosing the extractor $F$ to be some quantum-proof randomness extractor, such as those presented in~\cite{DPVR09}. The protocol uses approximately $(h(\gamma) + 2\gamma)n$ random bits, to decide about the testing rounds and to choose the inputs of the players on those rounds and $O(\log^3 n)$ random bits for the seed of the randomness extractor. By taking $\gamma = \Theta(\frac{\log n}{n})$ for instance we have used a polylogarithmic (in $n$) number of random bits and generated a linear number of bits $n$, thus achieving exponential randomness expansion. We refer the reader to~\cite{CVY13,MS14b} for further discussions on the way to generate the random bits needed for the protocol.
\begin{proof}
    We apply Theorem~\ref{thm:entropyaccumulationext} on $\rho$ with the substitutions $A_i \rightarrow A_i B_i$, $B_i \rightarrow T_i X_i Y_i$,  and $X_i \rightarrow C_i$,  where 
    \begin{align*}
     C_i = \begin{cases} \bot & \text{ if $T_i=0$}\\
        1 & \text{ if $T_i=1$ and $A_i \XOR B_i = X_i Y_i$}\\
        0 & \text{ if $T_i=1$ and $A_i \XOR B_i \neq X_i Y_i$.} \end{cases}
      \end{align*}
      Note that $C_i$ is a deterministic function of the classical registers $A_i B_i X_i Y_i T_i$ and the Markov conditions are clearly satisfied.
 
Note that the maps $\cM_i$ correspond to infrequent sampling maps with testing probability $\gamma$. As such, to compute a tradeoff function, we use the approach proposed in Lemma~\ref{lem:infrequent_sampling}. We start by determining a function $g : \mbP(\{0,1\}) \to \mbR$ satisfying the property~\eqref{eq:tradeoff_g}. Note that a distribution $q \in \mbP(\{0,1\})$ can be uniquely specified by $q(1) \in [0,1]$. For this reason, we will interpret $g$ as a function $g : [0,1] \to \mbR$. Note that the map $\cM_i^{\mathrm{test}}$ is of the form in Figure~\ref{fig:di-mi-circuit} except that $T$ is fixed to $1$ and thus $X$ and $Y$ are chosen uniformly at random, whereas $\cM_i^{\mathrm{data}}$ corresponds to $T$ being fixed to $0$. The inequality~\eqref{eqn:pabgs09} mentioned above shows that $g^*$ 
satisfies the property~\eqref{eq:tradeoff_g} when $q'(1) \in [\frac{3}{4}, \cos^2(\frac{\pi}{8})]$. 
However, $g^*$ is not an affine function. Nonetheless, $g^*$ is convex so any tangent provides a lower bound and also satisfies the property~\eqref{eq:tradeoff_g}. We consider the function obtained by taking the tangent at the point $p_b \in (\frac{3}{4}, \cos^2(\frac{\pi}{8}))$: for $p \in [0,1]$
\begin{align}
\label{eq:tangent} 
g_{p_b}(p) &= g^*(p_b) + (p-p_b) \frac{d g^*}{d \omega}(p_b) \ .
\end{align}
Note that $g_{p_b}$ satisfies property~\eqref{eq:tradeoff_g} required by Lemma~\ref{lem:infrequent_sampling}: when $p \in [\frac{3}{4}, \cos^2(\frac{\pi}{8})]$ it follows from the fact that $g_{p_b}(p) \leq g^*(p)$, when $p \in [0,\frac{3}{4}]$ from the fact that $g_{p_b}(p) \leqslant 0$ and recall that for $p > \cos^2(\pi/8)$, the right hand side of~\eqref{eq:tradeoff_g} is $+ \infty$. To get the bound stated in the theorem, we simply take $p_b = e$, but we note that choosing $p_b < e$ can lead to better bounds depending on the values of $\gamma$ and $n$. Note that $\Max{g_{p_b}} = g_{p_b}(1)$ and $\Min{g_{p_b}} = g_{p_b}(0)$. Applying Lemma~\ref{lem:infrequent_sampling}, we get a min-tradeoff function $f$ defined by $f(\delta_0) = g_{p_b}(1) + \frac{1}{\gamma}(g_{p_b}(0) - g_{p_b}(1)), f(\delta_1) = f(\delta_{\perp}) = g_{p_b}(1)$ and satisfies $\MinSigma{f} \geqslant g_{p_b}(0)$ and $\Var{f} \leq \frac{1}{\gamma}(g_{p_b}(1) - g_{p_b}(0))^2$.

As previously mentioned $\Omega$ is defined to be the event of not aborting, i.e. using the notation of the EAT we have
\[
\Omega = \{x_1^n \in \{0,1,\perp\}^n : |\{ i : x_i = 0 \}| \leq (1-e) \gamma n\}.
\] 
Observe that we have for $x_1^n \in \Omega$, 
\begin{align*}
f(\freq{x_1^n}) 
&= \freq{x_1^n}(0) f(\delta_0) + \freq{x_1^n}(1) f(\delta_1) + \freq{x_1^n}(\perp) f(\delta_{\perp}) \\
&= \freq{x_1^n}(0) \left(g_{p_b}(1) - \frac{1}{\gamma} (g_{p_b}(1) - g_{p_b}(0)) \right) + (1 -  \freq{x_1^n}(0)) g_{p_b}(1) \\
&\geqslant g_{p_b}(1) + (1-e) (g_{p_b}(0) - g_{p_b}(1)) \\
&= g_{p_b}(e) \ .
\end{align*}
Note that if $\Pr[\Omega] < \varepsilon$, then we are in case 2 of the theorem, so we will assume that $\rho[\Omega] = \Pr[\Omega] \geqslant\varepsilon$. Applying Theorem~\ref{thm:entropyaccumulationext}, we get
    \[ H_{\min}^{\varepsilon}(A_1^n B_1^n| E T_1^n X_1^n Y_1^n)_{\rho_{|\Omega}} > n g_{p_b}(e) - \sqrt{\frac{n}{\gamma}} c - c', \]
where $c = \sqrt{2 \ln 2} \left( \log 33 + \sqrt{2 + \frac{1}{\gamma} \left(\frac{dg^*}{d\omega}(p_b)\right)^2} \right)\sqrt{1 - 4 \log (\varepsilon)}$ and $c'$ is a constant only depending on $\varepsilon$ and $\frac{dg^*}{d\omega}(p_b)$.
\end{proof}

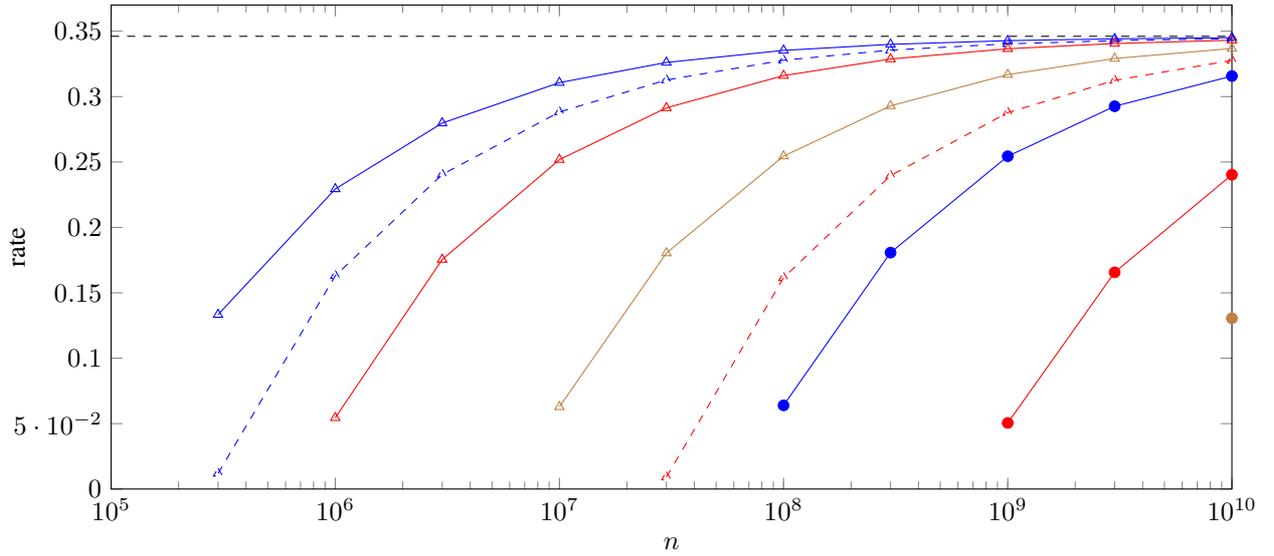
\begin{figure}
    \begin{center}
    \begin{tikzpicture}
        \begin{semilogxaxis}[
            xlabel=$n$,
            ylabel=rate,
            xmin=1e5,
            xmax=1e10, 
            ymin=0, 
            ymax=0.37, 
            width=0.9\textwidth, 
            height=8cm,
        ]
            \addplot[blue, mark=triangle] table {rate-plots-0.dat}; 
            \addplot[red, mark=triangle] table {rate-plots-1.dat}; 
            \addplot[brown, mark=triangle] table {rate-plots-2.dat}; 
            \addplot[blue, mark=*] table {rate-plots-3.dat}; 
            \addplot[red, mark=*] table {rate-plots-4.dat}; 
            \addplot[brown, mark=*] table {rate-plots-5.dat}; 

            \addplot[dashed, blue, mark=triangle] table {rate-plots-arv-0.dat}; 
            \addplot[dashed, red, mark=triangle] table {rate-plots-arv-1.dat}; 

            \draw (axis cs:100000, 0.3461) edge[dashed] (axis cs:1e10, 0.3461);
    \end{semilogxaxis}
    \end{tikzpicture}
    \end{center}
    \caption{Plot of the final entropy rate achieved as a function of the number of rounds $n$ for several values of the sampling probability $\gamma$, when fixing the winning threshold $e$ to 0.8, for both this paper (solid lines) and using the blocking technique of~\cite[Appendix B, rates given by equation (36)]{arv16} (dashed lines). In decreasing order, we have $\gamma=1$, $\gamma=0.1$, $\gamma=0.01$, $\gamma=0.001$, $\gamma=0.0001$, and the last point is $\gamma=3 \times 10^{-5}$. We point out that the blocking technique of~\cite{arv16} only gives positive rates in this regime when $\gamma = 1$ and $\gamma = 0.1$. Here, we fixed $\varepsilon=10^{-5}$ and assumed that $\rho[\Omega] \geqslant 10^{-5}$. The black dashed line at the top corresponds to the first-order rate of roughly 0.3461, i.e.~when $n \rightarrow \infty$. }
\label{fig:rate-plots}
\end{figure}

\section{Conclusion and open problems}\label{sec:conclusion}
The new version of the entropy accumulation theorem presented here can now be applied directly to protocols with infrequent sampling (or to other situations where the entropy variance is significantly different from the local dimension) without paying too heavy a price. In particular, in the infrequent sampling case, the scaling in the sampling frequency $\gamma$ roughly matches what we would expect in the classical i.i.d.~case from Chernoff-type bounds.

As we noted earlier, another way to obtain this scaling in the infrequent sampling case is by the blocking technique used in~\cite[Appendix B]{arv16}: instead of applying the (original) EAT to individual rounds, they apply it to blocks of size $O(\frac{1}{\gamma})$, which ensures that each block has roughly one test round. By doing this, one gets a tradeoff function (which now acts on blocks) whose gradient scales correctly. There are multiple advantages of our method over this technique. First, it is more general, since it can gracefully handle cases beyond infrequent sampling, where $\Var{f}$ is substantially different from the local dimension for other reasons. It is also more natural and simpler to use, there is no need to handle the additional parameters related to the blocking. Furthermore, it appears to lead to significantly better bounds: the numerical results we obtain for the protocol given in Section~\ref{sec:DIRE} are substantially better (see the dashed lines in Figure~\ref{fig:rate-plots} for the bounds obtained using the blocking method), and there is no particular reason to think that this case is not representative. 

While the results given here are largely good enough in practice, there are still open questions remaining. First: can we find a version of the theorem with an optimal second-order term? Ideally, we could hope for a second-order term that matches what we see in the i.i.d.~case, which would look like $\sqrt{n V} \Phi^{-1}(\varepsilon^2)$ (e.g., in~\cite{th12}), where $\Phi$ is the cumulative distribution function of a Gaussian distribution, and $V$ would be an appropriate entropy variance term. Here we fall short of this in two ways: first, our $V$ quantity is the result of applying some inequalities in the proof (see Equation \eqref{eqn:V-inequality}) that are not always tight; this may however be unavoidable if one wants to have a clean expression in terms of $\Var{f}$. The second issue is that the dependence in $\varepsilon$ does not match the $\Phi^{-1}(\varepsilon)$ or $\Phi^{-1}(\varepsilon^2)$ that is usually seen in second-order expansions, but is instead similar to what is done in the fully quantum AEP of \cite{TCR09} and in the original EAT. This also seems very difficult to overcome in our situation, since these terms usually arise from an application of the Berry-Esseen theorem, which quantifies how much a sum of iid random variables diverges from a normal distribution, and therefore depends very strongly on the iid assumption which we do not have here.

We could also scale back our goals a bit and try to improve the last term, namely the $c'$ in Equation~\eqref{eqn:eat-min}. As it stands, this term arises from a sequence of ad-hoc inequalities that could very well be improved. It would be particularly interesting to understand which parameters this term should ``really'' depend on: for example, the expression we give here depends on $\rho[\Omega]$ and $\varepsilon$, but this may well be an artifact of our choice of $\alpha$ in the proof. We thus leave these questions as open problems.

\appendix

\section{Various lemmas}

\begin{lemma}[Lemma B.5 in \cite{dfr16}]\label{lem:abx-chain-rule-opt}
    Let $\rho_{AB}$ be a quantum state of the form $\rho = \sum_x p_x {\rho}_{AB|x}$, where $\{ p_x \}$ is a probability distribution over $\mathcal{X}$. Then, for any $x \in \mathcal{X}$ and any $\alpha \in (1, \infty)$,
    \begin{align}
        \label{eqn:abx-chain-rule-gtr1-opt} H^{\uparrow}_\alpha(A|B)_{\rho} - \frac{\alpha}{\alpha - 1} \log \left( \frac{1}{p_x} \right)  &\leqslant  H^{\uparrow}_{\alpha}(A|B)_{\rho_{|x}} \ .
    \end{align}
    and for $\alpha \in (0,1)$, 
        \begin{align}
        \label{eqn:abx-chain-rule-gtr1-opt-lessthan1} H^{\uparrow}_\alpha(A|B)_{\rho} - \frac{\alpha}{\alpha - 1} \log \left( \frac{1}{p_x} \right)  &\geqslant  H^{\uparrow}_{\alpha}(A|B)_{\rho_{|x}} \ .
    \end{align}
\end{lemma}

\begin{lemma}\label{lem:hprimealpha-part-classical-dim-bound}
    Let $\rho_{ABX} = \sum_x p_x \rho_{AB}(x) \otimes \proj{x}_X$ be a quantum state with $X$ classical. Then, for any $\alpha \in (0, 1) \cup (1, 2]$, we have
    \[ H'_{\alpha}(AX|B) \geqslant -\log d_A. \]
\end{lemma}
\begin{proof}
    First, let us define the extension $\rho_{ABXX'} = \sum_{x} p_x \rho_{AB}(x) \otimes \proj{xx}_{XX'}$, and observe that by data processing, $H'_{\alpha}(AX|B) \geqslant H'_{\alpha}(AX|BX')$. Thus,
    \[ \rho_{ABXX'}^{\alpha} \rho_{BX'}^{1-\alpha} = \sum_x p_x \rho_{AB}(x)^{\alpha} \rho_B(x)^{1-\alpha} \otimes \proj{xx}_{XX'} \]
    and therefore
    \begin{align*}
        H'_{\alpha}(AX|B) &\geqslant H'_{\alpha}(AX|BX')\\
        &= \frac{1}{1-\alpha} \log \sum_x p_x 2^{(1-\alpha) H'_{\alpha}(A|B,X=x)_{\rho}}\\
        &\geqslant \frac{1}{1-\alpha} \log \sum_x p_x 2^{(\alpha-1) \log d_A}\\
        &\geqslant -\log d_A.
    \end{align*}
\end{proof}

\section*{Acknowledgments}
The authors would like to thank Rotem Arnon-Friedman for bringing to our attention the suboptimal dependence on the testing probability in the original entropy accumulation theorem and Renato Renner for his comments. We would also like the thank the IEEE Transactions on Information Theory reviewers for their detailed feedback on the manuscript. This work is supported by the French ANR project ANR-18-CE47-0011 (ACOM).

\printbibliography

\end{document}